\documentclass[reqno]{amsart}
\usepackage{xcolor}
\usepackage{amsmath, amssymb}
\usepackage{latexsym}
\usepackage{wasysym}
\usepackage{tikz}
\usepackage{graphicx}
\usepackage[pdfborder={0 0 0},
colorlinks=true,linkcolor=blue,citecolor=red,anchorcolor=green,CJKbookmarks=true]{hyperref}
\allowdisplaybreaks
\newtheorem{Theorem}{Theorem}[section]
\newtheorem{lemma}[Theorem]{Lemma}

\newtheorem{Corollary}[Theorem]{Corollary}
\newtheorem{remark}[Theorem]{Remark}

\numberwithin{equation}{section}

\newcommand{\la}{\langle}
\newcommand{\ra}{\rangle}

\newcommand{\beq}{\begin{equation}}
\newcommand{\eeq}{\end{equation}}
\newcommand{\bes}{\begin{equation*}}
\newcommand{\ees}{\end{equation*}}

\renewcommand{\l}{\ell}

\newcommand{\dss}{\displaystyle}

\newcommand{\bslj}{J_{\nu}}
\newcommand{\bsljm}{J_{-\nu}}
\newcommand{\bslh}{H_{\nu}^{(1)}}

\newcommand{\belh}{H_{\frac{1}{3}}^{(1)}}

\newcommand{\sbelj}[1]{\sqrt{\frac{\pi#1}{2}}J_{\frac{1}{3}}(#1)}
\newcommand{\sbeljm}[1]{\sqrt{\frac{\pi#1}{2}}J_{-\frac{1}{3}}(#1)}

\newcommand{\sbelhz}[1]{\sqrt{\frac{\pi z}{2}}H_{\frac{1}{3}}^{(1)}(#1)}
\newcommand{\lrp}[1]{\ensuremath{\left(#1\right)}}
\newcommand{\Bgp}[1]{\ensuremath{\Big(#1\Big)}}

\newcommand{\Bgs}[1]{\ensuremath{\Big|#1\Big|}}

\newcommand{\gsf}{\Gamma\lrp{\frac{5}{6}}}
\newcommand{\gso}{\Gamma\lrp{\frac{1}{6}}}

\def\e{\epsilon}

\def\N{{\Bbb N}}
\def\Z{{\Bbb Z}}
\def\R{{\Bbb R}}
\def\T{{\Bbb T}}
\def\C{{\Bbb C}}

\let\cal=\mathcal
\def\H{{\cal H}}

\hfuzz=\maxdimen
\def\e{\epsilon}

\def\N{{\Bbb N}}
\def\Z{{\Bbb Z}}
\def\R{{\Bbb R}}
\def\T{{\Bbb T}}
\def\C{{\Bbb C}}

\let\cal=\mathcal

\hfuzz=\maxdimen
\begin{document}
\newtheorem{Sup}{\textbf{Assumption}}

\title[A Reducibility Theorem]{Reducibility of 1-d Quantum Harmonic Oscillator Equation with  Unbounded Oscillation  
Perturbations }
\author{Z. Liang and J. Luo}

\address {School of Mathematical Sciences and
Key Lab of Mathematics for Nonlinear Science, Fudan University,
Shanghai 200433, China} \email{zgliang@fudan.edu.cn, 17210180007@fudan.edu.cn}

\date{}

\begin{abstract} We build a new estimate relative with Hermite functions based 
upon oscillatory integrals and Langer's  turning point theory.  
From it we show that the equation
\begin{eqnarray*}
{\rm i}\partial_t u =-\partial_x^2 u+x^2 u+\e \la x\ra^{\mu} W(\nu x,\omega t)u,\quad u=u(t,x),~x\in\R,~ 0\leq \mu<\frac13,
 \end{eqnarray*}
can be reduced in $\cal H^1(\R)$  to an autonomous system for most values of the frequency vector $\omega$ and $\nu$,
where  $W(\varphi, \theta)$  is a smooth map from $ \T^d\times \T^n$ to $\R$ and odd in $\varphi$.

\end{abstract}

\maketitle

\section{Introduction of the main results }\label{introduction}
\subsection{Statement of the Results}
 \par In this paper we consider 1-d quantum harmonic oscillator equation
\begin{eqnarray}\label{maineq}
{\rm i}\partial_tu &=& H_{\varepsilon}(\omega t)u, \ x\in\R,\\
H_{\varepsilon}(\omega t) : &= & -\partial_{xx}+x^2+\varepsilon \la x\ra ^\mu W(\nu x, \omega t), \nonumber
 \end{eqnarray}
where $0\leq \mu<\frac13$, $W(\varphi, \theta)$ is defined on $\T^d\times \T^n$ and satisfies
\begin{eqnarray}\label{symmetry}
W(-\varphi, \theta)=-W(\varphi, \theta), \qquad {\rm for}\ \forall\  (\varphi, \theta)\in \T^{d}\times \T^n.
\end{eqnarray}
In order to state the results we need to introduce  some notations and spaces.\\
{\bf $C^{\beta}(\R^n, X)$}. Assume that $X$ is a complex Banach space with the norm $\|\cdot \|_{X}$. Let $\mathcal{C}^{b}(\R^n,X)$, $0<b<1$, be the space of
H\"older continuous functions $f : \R^n\rightarrow X$ with the norm
$$\|f\|_{\mathcal{C}^{b}(\R^n, X)} : = \sup\limits_{0<\|z_1-z_2\|<2\pi}\frac{\|f(z_1)-f(z_2)\|_{X}}{{\|z_1-z_2\|^{b}}}+\sup\limits_{z\in \R^n} \|f(z)\|_{X}.$$
If $b=0$, then $\|f\|_{\mathcal{C}^{b}(\R^n, X)}$ denotes the sup-norm. For $\beta=[\beta]+b$ with $0\leq b<1$, we denote
by ${\mathcal{C}^{\beta}}(\R^n, X)$ the space of functions $f: \R^n\rightarrow X$ with H\"older continuous partial derivatives and $\partial^{\nu} f\in \mathcal{C}^{b}(\R^n, X_{\nu})$ for all
multi - indices $\nu=(\nu_1, \cdots, \nu_n)\in \N^n$, where $|\nu| : = |\nu_1|+\cdots+|\nu_n| \leq \beta$ and $X_{\nu}=\mathfrak{L}(\prod\limits_{i=1}^{|\nu|}Y_i, X)$ with the standard norm and $Y_i : =\R^n$, $i=1, \cdots, |\nu|$.
We define the norm
$
\|f\|_{\mathcal{C}^{\beta}(\R^n, X)} := \sum\limits_{|\nu|\leq \beta}\|\partial^{\nu}f\|_{\mathcal{C}^{b}(\R^n, X_{\nu})}.
$ \\
{\bf $C^{\beta}(\T^n, X)$}. Denote by $\mathcal{C}^{\beta}(\T^n, X)$ the space of all functions $f\in \mathcal{C}^{\beta}(\R^n, X)$ that are of period $2\pi$ in all variables. We define
$\|f\|_{\mathcal{C}^{\beta}(\T^n, X)} :=  \|f\|_{\mathcal{C}^{\beta}(\R^n, X)} $.\\
{\bf Linear Space.} Let $s\in \R$, we define the complex weighted-$\l^2$-space
\[
\l_s^2=\big\{\xi=(\xi_j\in\C,~j\in\Z_+)~\big|~\|\xi\|_s<\infty\big\},~\text{where }~\dss\|\xi\|_s^2=\sum_{j\in\Z_+}j^s|\xi_j|^2.
\]
{\bf Hermite functions} The harmonic oscillator operator $T=-\frac{d^2}{dx^2}+x^2$  has eigenfunctions $(h_j)_{j\geq 1}$, so called  the Hermite functions, namely, 
\begin{eqnarray}\label{eigensection1}
Th_j=(2j-1)h_j, \qquad \|h_j\|_{L^2(\R)}=1,\qquad  j\geq 1.
\end{eqnarray} 
{\bf $\mathcal{H}^p$. } Let $p\geq 0$ be an integer we define
$$ \mathcal{H}^p:=\{u\in \mathcal{H}^p(\R,\C)\ |\  x\mapsto x^{\alpha_1}\partial_x^{\alpha_2}u \in L^{2}(\R)\ {\rm for\ any}\ \alpha_1,\alpha_2\in\N,\     0\leq \alpha_1+\alpha_2\leq
p \}.$$
\noindent To a function $u\in\mathcal{H}^p$ we associate the sequence $(u_j)_{j\geq 1}$ of its Hermite coefficients by the formula $u(x)=\sum_{j\geq 1}u_jh_{j}(x).$
For $p\geq0$,
$u\in\mathcal{H}^p\Leftrightarrow (u_j)_{j\geq 1}\in\ell_{p}^2$ and we define its  norm by 
$$\|u\|_p=\|(u_j)_{j\geq 1}\|_p=(\sum_{j\geq 1}j^p|u_j|^2)^{\frac12}. $$
For simplicity  we define $\alpha=\frac{1}{12}-\frac{\mu}{4}$ and $\beta_{*}(n,\mu) =18(n+3)(2+\alpha^{-1})(2+5\alpha^{-1})$.  Our intent is to prove the following
\begin{Theorem}\label{quantumth}
Assume  that $W(\varphi, \theta)$ satisfies (\ref{symmetry}) and $W(\varphi, \theta)$ is $C^{s}(\T^d\times \T^n)$ with
$s\geq d+[1\vee \tau]+n+3+\beta$ and $\beta>\beta_{*}(n,\mu) $ with $\beta\in \Z$ and $\tau>d-1$. 
There exists $\e_*>0$ such that for all $0\leq \e<\e_*$ there exists a closed set $\Omega_\gamma\times D_\varepsilon\subset [A,B]^d\times [1,2]^n$ 
and for all $(\nu, \omega)\in \Omega_\gamma\times D_\varepsilon$ the linear Schr\"odinger equation  (\ref{maineq})
reduces to a linear autonomous equation in the space $\mathcal{H}^{1}$.\\
\indent More precisely,  there exists $\e_*>0$ such that for all $0\leq \e<\e_*$ there exists a closed set $\Omega_{\gamma}\times D_\varepsilon$,  and for $(\nu,\omega) \in \Omega_{\gamma}\times D_\varepsilon$, 
there exist a linear isomorphism $\Psi_\omega^{\infty,1}(\theta)\in\mathfrak{L}(\mathcal{H}^{s'})$ with $0\leq s'\leq 1$,  unitary on $L^2(\R)$, where  $\Psi_\omega^{\infty,1}(\theta)- id \in \mathcal{C}^{\iota}(\T^n,\mathfrak{L}(\mathcal{H}^{0}, \mathcal{H}^{2\alpha})) \cap
\mathcal{C}^{\iota}(\T^n,\mathfrak{L}(\mathcal{H}^{s'}, \mathcal{H}^{s'}))$  with $\iota\notin\Z$ and $\iota\leq\frac{2}{9}\beta$ and a bounded Hermitian operator $Q_1\in\mathfrak{L}(\mathcal{H}^{1})$ such that $t\mapsto u(t,\cdot)\in\mathcal{H}^{1}$ satisfies (\ref{maineq}) if and only if $t\mapsto v(t,\cdot)= \Psi_\omega^{\infty,1}(\omega t)u(t,\cdot) $ satisfies the autonomous equation
\begin{equation*}\label{reducedeq}
\mathrm{i}\partial_t v= -v_{xx}+x^2 v+\varepsilon Q_1(v),
\end{equation*}
furthermore,
$$\|\Psi_\omega^{\infty,1}(\theta)-id\|_{\mathcal{C}^{\iota}(\T^n, \mathfrak{L}(\mathcal{H}^{0},\mathcal{H}^{2\alpha}))}\leq C \varepsilon^{\frac{3}{2\beta}(\frac{2}{9}\beta-\iota)},\  (\theta,\nu, \omega)\in\T^n\times \Omega_{\gamma}\times D_\varepsilon,$$
and
$$\|\Psi_\omega^{\infty,1}(\theta)-id\|_{\mathcal{C}^{\iota}(\T^n, \mathfrak{L}(\mathcal{H}^{s'},\mathcal{H}^{s'}))}\leq C \varepsilon^{\frac{3}{2\beta}(\frac{2}{9}\beta-\iota)},\  (\theta,\nu, \omega)\in\T^n\times \Omega_{\gamma}\times D_\varepsilon,$$
for $0\leq s'\leq 1$,
and for any $p\in \N$ and $\omega\in D_{\varepsilon}$, there exists a constant $K_1$ depending on $n, \beta$
\begin{eqnarray}\label{Whitneysmooth1}
\|Q_1\|_{\mathcal{L}(\mathcal{H}^p, \mathcal{H}^{p+4\alpha})}+ \|\partial_{\omega}Q_1\|_{\mathcal{L}(\mathcal{H}^p, \mathcal{H}^{p+4\alpha})}\leq K_1.
\end{eqnarray}
\end{Theorem}

\begin{remark}
The sets $\Omega_{\gamma}$  satisfies
${\rm Meas}([A,B]^d\setminus \Omega_\gamma)= \mathcal{O}(\gamma)$ when $\gamma\rightarrow 0$, while $D_\varepsilon$ 
satisfies ${\rm Meas}([1,2]^n\setminus D_\varepsilon)\leq c(\beta,n,\mu)\varepsilon^{\frac{3/2}{(2+\alpha^{-1})(2+5\alpha^{-1})}}$.
\end{remark}
\begin{remark}\label{parameternv}
For any $\nu\in \Omega_{\gamma}\subset [A, B]^d$, we have
$|\la k, \nu\ra|\geq \frac{\gamma}{|k|^{\tau}}$ where $k\neq 0$ and $\tau>d-1$.
\end{remark}
\begin{remark}
The derivative in (\ref{Whitneysmooth1}) is in the sense of Whitney.
\end{remark}
\indent A consequence of the above theorems and corollary is that in the considered range of parameters all the Sobolev norms, i.e. the $\H^s$ norms of the solutions are bounded forever and the spectrum of the Floquet operator is pure point.\\
\indent  Consider 1-d quantum harmonic oscillator equation
\begin{eqnarray}\label{maineq2}
{\rm i}\partial_t u &=& \mathcal H_{\varepsilon}(\omega t)u, \ x\in\R,\\
\mathcal H_{\varepsilon}(\omega t) : &= & -\partial_{xx}+x^2+\varepsilon X(x, \omega t), \nonumber
 \end{eqnarray}
where
\begin{eqnarray}\label{realform1}
X(x,\theta)= \la x\ra^{\mu}  \sum\limits_{k\in \Lambda} (a_k(\theta)\sin kx+b_k(\theta) \cos kx)
\end{eqnarray}
 with $k\in \Lambda\subset \R \setminus \{0\}$ with $|\Lambda|<\infty$ and $0\leq \mu<\frac{1}{3}$.
\begin{Corollary}\label{coro1.5}
Assume that $a_k(\theta)$ and $b_k(\theta)\in C^{r}(\T^n)$ where $r\ge n+2+\beta$ with $\beta$ as in Theorem \ref{quantumth}.
There exists $\e_*>0$ such that for all $0\leq \e<\e_*$ there exists a closed set $D_\varepsilon\subset [1,2]^n$  such that for
 all $\omega\in D_\varepsilon$ the linear Schr\"odinger equation  (\ref{maineq2})
reduces to a linear autonomous equation in the space $\mathcal{H}^{1}$.\\ 
 \indent  More precisely,  there exists $\e_*>0$ such that for all $0\leq \e<\e_*$ there exists a closed set $D_\varepsilon$
 and for $\omega\in D_\varepsilon$, there exist a linear isomorphism $\Psi_\omega^{\infty,2}(\theta)\in\mathfrak{L}(\mathcal{H}^{s'})$ with $0\leq s'\leq 1$,  unitary on $L^2(\R)$, where  $\Psi_\omega^{\infty,2}(\theta)- id \in \mathcal{C}^{\iota}(\T^n,\mathfrak{L}(\mathcal{H}^{0}, \mathcal{H}^{2\alpha})) \cap
\mathcal{C}^{\iota}(\T^n,\mathfrak{L}(\mathcal{H}^{s'}, \mathcal{H}^{s'}))$  with $\iota\notin\Z$ and $\iota\leq\frac{2}{9}\beta$ and a bounded Hermitian operator $Q_2\in\mathfrak{L}(\mathcal{H}^{1})$ such that $t\mapsto u(t,\cdot)\in\mathcal{H}^{1}$ satisfies (\ref{maineq}) if and only if $t\mapsto v(t,\cdot)= \Psi_\omega^{\infty,2}(\omega t)u(t,\cdot) $ satisfies the autonomous equation
\begin{equation*}\label{reducedeq}
\mathrm{i}\partial_t v= -v_{xx}+x^2 v+\varepsilon Q_2(v).
\end{equation*}
Furthermore,
$$\|\Psi_\omega^{\infty,2}(\theta)-id\|_{\mathcal{C}^{\iota}(\T^n, \mathfrak{L}(\mathcal{H}^{0},\mathcal{H}^{2\alpha}))}\leq C \varepsilon^{\frac{3}{2\beta}(\frac{2}{9}\beta-\iota)},\  (\theta,\omega)\in\T^n\times D_\varepsilon,$$
and
$$\|\Psi_\omega^{\infty,2}(\theta)-id\|_{\mathcal{C}^{\iota}(\T^n, \mathfrak{L}(\mathcal{H}^{s'},\mathcal{H}^{s'}))}\leq C \varepsilon^{\frac{3}{2\beta}(\frac{2}{9}\beta-\iota)},\  (\theta,\omega)\in\T^n\times D_\varepsilon,$$
for $0\leq s'\leq 1$,
and for any $p\in \N$ and $\omega\in D_{\varepsilon}$, there exists a constant $K_2$ depending on $n, \beta$,
\begin{eqnarray}\label{Whitneysmooth2}
\|Q_2\|_{\mathcal{L}(\mathcal{H}^p, \mathcal{H}^{p+4\alpha})}+ \|\partial_{\omega}Q_2\|_{\mathcal{L}(\mathcal{H}^p, \mathcal{H}^{p+4\alpha})}\leq K_2.
\end{eqnarray}
\end{Corollary}
\begin{remark}
The set $D_\varepsilon$ satisfies 
${\rm Meas}([1,2]^n\setminus D_\varepsilon)\leq c(\beta,n,\mu)\varepsilon^{\frac{3/2}{(2+\alpha^{-1})(2+5\alpha^{-1})}}$.
\end{remark}

\subsection{Related results and a critical lemma}
\indent  In the following we recall some relevant results. For 1d harmonic oscillator see  \cite{Com87} and  \cite{EV}  for periodic in time bounded perturbations.  
Refer to  \cite{GT11}, \cite{Wang08} and  \cite{WLiang17} for 1d harmonic oscillators with quasi-periodic in time bounded perturbations.\\
\indent In \cite{BG} Bambusi and Graffi first proved the reducibility of 1d Schr\"odinger equation with an unbounded time quasiperiodic perturbation. In \cite{BG} they assumed that the potential grows at infinity like $|x|^{2l}$  with  a real $l>1$ and the perturbation is bounded by $1+|x|^{\beta}$ with $\beta<l-1$;  reducibility in the limiting case $\beta= l-1$ was obtained by Liu and Yuan in \cite{LY10}. Recently, the results in \cite{BG} and \cite{LY10} have been improved in \cite{BamII, BamI}, in which Bambusi firstly obtained the reducibility results for 1d harmonic oscillators with unbounded perturbations.  In  \cite{BamI}  Bambusi proved the reducibility when the symbol of the perturbation grows at most like $(\xi^2+x^2)^{\beta/2}$ with $\beta<2$.  In \cite{BamII} he generalized the class of the symbol to which the perturbation belongs(see \cite{BamIII}).  In remark 2.7 \cite{BamII}, Bambusi wrote ``we also remark that the assumption that the functions $a_i$ are symbols rules out cases like $a_i(x, \omega t) = \cos(x -\omega t)$.''  The terms ``$a_i$''  are  exactly the oscillatory ones considered in this paper. \\
 \indent More applications of pseudodifferential calculus can be found in several papers (see e.g. \cite{BBM14, BBHM17, BM16, FP15, Giu17, IPT05,Mon17a,PT01}).  We mention that the above results  are limited in the one dimensional case, while some higher dimensional results on this problems have been recently obtained \cite{BGMR18,FGMP18, GP16, LiangW19,  Mon17b}. The related techniques have been used for a control on the growth of Sobolev norms in \cite{BGMR17, Mon18}.\\
\indent The proof of Theorem \ref{quantumth} is based upon the KAM in \cite{LiangW19} and the following  estimate of Hermite functions.  
\begin{lemma}\label{Indecaysection1}
Suppose $h_m(x)$ satisfies (\ref{eigensection1}). For any $k\neq 0$ and  for any $m,n\geq 1$,
$$\big|\int_{\R}\la x\ra^{\mu} e^{{\rm i} kx}h_m(x) \overline{{h}_n(x)}dx\big| \leq C (|k|\vee  |k|^{-1})m^{-\frac{1}{12}+\frac{\mu}{4}} n^{-\frac{1}{12}+\frac{\mu}{4}}, $$
where $C$ is an absolute constant and $0\leq \mu<\frac13$.
\end{lemma}
\begin{remark}
We guess the estimate in Lemma \ref{Indecaysection1} is optimal and will give a numerical proof for it in the following paper.  
\end{remark}
\indent By Theorem \ref{KAM} in sect.  \ref{s3} and Lemma \ref{Indecaysection1} we prove Thm. \ref{quantumth} in section \ref{s4}.   
For the readers' convenience we give a fast introduction of 
Langer's turning point theory  from \cite{T2} at the beginning of section \ref{S2}. The lengthy proof for  Lemma \ref{Indecaysection1}  is then given after it.  
Section \ref{S5} is divided into two parts. In the first part we present Theorem \ref{maintheorem2} without proof. In the second we give some lemmas. \\
\noindent {\it Notations}.   
For $k\in \Z^n$, $|k|=\sum\limits_{j=1}^n|k_j|$. We use $\la x\ra=\sqrt{1+x^2}$.
$\la\cdot ,\cdot \ra$ is the standard scalar product in $\R^n$ or $\ell^2$. $\|\cdot \|$ is an operator-norm or $\ell^2$-norm.  We use the notations $1\vee \tau: = \max \{1, \tau\}$, $\Z_{+}=\{1,2,\cdots\}$ and $\N=\{0,1,2,\cdots\}$ and $\T=\R/ 2\pi\Z$.
For a positive number $a$, $[a]$ means the largest integer not larger than $a$.  
We use the notation $f(x)=O(g(x))$ for $x\rightarrow \infty$ if $|f(x)|\leq C|g(x)|$ when $x$ is large enough.  We denote $ \T^n_{\rho} =\big\{(a+{\rm i}b)\in \C^n/{2\pi \Z^n}\big| \max\limits_{j} |b_j|<\rho\big\}$. The notation $Meas$ stands for the Lebesgue measure in $\R^n$. \\
{\bf Acknowledgements.} The first author is very grateful to Bambusi, D.  for many  invaluable discussions on this paper.  Both authors were partially supported by National Natural Science Foundation of China(Grants No.  11371097; 11571249) and 
Natural Science Foundation of Shanghai(Grants No. 19ZR1402400).

\section{A KAM Theorem}\label{s3}
We introduce the KAM Theorem from \cite{LiangW19} especially for 1d case.  We remark that KAM theory is almost well-developed for nonlinear Hamiltonian PDEs in 1-d context. See \cite{BBP2, GY1, KLiang, KaPo, KP,  Kuk93, Ku1, Ku2, LZ, LY1, LiuYuan, Pos, W90, ZGY} for 1-d KAM results.  Comparing with 1-d case, the KAM results for multidimensional PDEs are relatively few. Refer to  \cite{EGK, EK, GXY, GY2, GP16,  PX} for n-d results. See \cite{Berti} and \cite{Berti2019} for an almost complete picture of recent KAM theory.\\
\subsection{Setting}\label{phase}
\noindent\emph{Linear space.} For $p\geq0$  we define
 $Y_p:=\ell_{p}^2\times\ell_{p}^2=\{\zeta=(\zeta_a=(\xi_a,\eta_a)\in\C^2)_{a\in\Z_{+}}\big|\ \|\zeta\|_p<\infty\}$
with  $\|\zeta\|_p^2:=\sum\limits_{a\in \Z_{+}} (|\xi_a|^2+|\eta_a|^2)a^{p}$. We provide the space $Y_p,\ p\geq0,$ with the symplectic structure $\mathrm{i}\sum_{a\in \Z_{+}}d \xi_a\wedge d \eta_a$. \\
\noindent\emph{Infinite matrices.} We denote by $\mathcal{M}_{\alpha}$ the set of infinite matrices $A:\mathcal{E}\times\mathcal{E}\rightarrow \C$ with the norm
$
|A|_{\alpha}:=\sup_{a,b\in \Z_{+}} \left(ab\right)^{\alpha}\big|A_a^b\big|<+\infty.
$
We also denote $\mathcal{M}_{\alpha}^+$ be the subspace of $\mathcal{M}_{\alpha}$ satisfying that an infinite matrix $A\in\mathcal{M}_{\alpha}^+$ if
$
|A|_{\alpha+}:=\sup_{a,b\in \Z_{+}}(a b)^{\alpha}\left(1+|a-b|\right)  |A_{a}^{b}|<+\infty.
$\\
\noindent\emph{Parameter.} In  the paper $\omega$ will play the role of a parameter belonging to $D_0=[1,2]^n$.
All the constructed functions will depend on $\omega$ with $\mathcal{C}^1$ regularity. When a function is only defined on a Cantor subset of $D_0$ the regularity
is understood in  Whitney sense.\\
\noindent\emph{A class of quadratic Hamiltonians.} Let $D\subset D_0, \ \alpha>0$ and $\sigma>0$. We denote by $\mathcal{M}_{\alpha}(D,\sigma)$ the set of
mappings  as
$\T^n_\sigma\times D\ni (\theta,\omega)\mapsto Q(\theta,\omega)\in \mathcal{M}_{\alpha}$
which is real analytic on $\theta\in \T^n_\sigma$ and $\mathcal{C}^1$
continuous on  $\omega\in D$. This space is equipped with the norm
$$[Q]_{\alpha}^{D,\sigma}:=\sup_{\substack{\omega\in D,|\Im \theta|<\sigma,\\|k|=0,1}}\left|\partial^k_\omega Q(\theta,\omega)\right|_{\alpha}.$$
\indent The subspace of $ \mathcal{M}_{\alpha}(D,\sigma)$ formed by  $F(\theta,\omega)$ such that $\partial^k_\omega F(\theta,\omega)\in \mathcal{M}^+_{\alpha},\ |k|=0,1,$ is denoted by $ \mathcal{M}_{\alpha}^+(D,\sigma)$ and
equipped with the norm
$[F]_{\alpha+}^{D,\sigma}:=\sup\limits_{\substack{\omega\in D,|\Im \theta|<\sigma,\\|k|=0,1}}\left|\partial^k_\omega F(\theta,\omega)\right|_{\alpha+}.$
The subspace of  $  \mathcal{M}_{\alpha}(D,\sigma)$ that are independent of $\theta$ will be denoted by $\mathcal{M}_{\alpha}(D)$ and for $N\in \mathcal{M}_{\alpha}(D),$
$$[N]_{{\alpha}}^{D}:=\sup_{\omega\in D,|k|=0,1}|\partial^k_\omega N(\omega)|_{{\alpha}}.$$
\noindent\emph{$\mathcal{C}^1$ norm of operator in $\omega$.} Given   $(\theta,\omega)\in\T_{\sigma}^n\times D$,  $\Phi(\theta,\omega)\in\mathfrak{L}(Y_r,Y_{r'})$ being $\mathcal{C}^1$
operator with respect to  $\omega$  in Whitney sense, we define the $C^1$ norm of $\Phi(\theta,\omega)$ with respect to $\omega$ by
$$\|\Phi\|^*_{\mathfrak{L}(Y_r,Y_{r'})}=\sup_{\substack{(\theta,\omega)\in\T_{\sigma}^n\times D,\\|k|=0,1,\ \|\zeta\|_s\neq0}}\frac{\|\partial_\omega^k\Phi (\theta,\omega)\zeta\|_{r'}}{\|\zeta\|_r},$$
where $r, r'\in \R$.
 \subsection{The reducibility theorem}\label{s2.4}
 \noindent In this subsection we state an abstract reducibility theorem for quadratic $t$-quasiperiodic   Hamiltonian of the form
\begin{equation}\label{hameq}
H(t,\xi,\eta)= \la\xi, N\eta\ra+ \varepsilon\la\xi, P(\omega t)\eta\ra, \quad (\xi,\eta)\in Y_1\subset Y_0,
\end{equation}
 and the associated Hamiltonian system  is
\begin{eqnarray*}
\left\{\begin{array}{c}
\dot{\xi}=-\mathrm{i}N\xi-\mathrm{i}\varepsilon P^T(\omega t)\xi,\\
\dot{\eta}=\ \ \mathrm{i}N\eta+\mathrm{i}\varepsilon P (\omega t)\eta,\
\end{array}\right.\label{hameq00}
\end{eqnarray*}
where $N=diag\{\lambda_a,\ a\in \Z_{+}\}$ %and $P(\theta)=(P_{a}^b(\theta))_{a,b\in\mathcal{E}}$ belongs to $\mathcal{C}^\beta(\T^n,\mathcal{M}_{s,\alpha})$
satisfying the following assumptions:\\
 \textbf{Hypothesis A1 - Asymptotics.} There exist   positive constants $c_0,\ c_1,\ c_2$ such that $$c_1 a\geq\lambda_a\geq c_2a
 \ {\rm and}\ |\lambda_a-\lambda_b|\geq c_0|a-b|,\ a,b\in \Z_{+}.$$
\textbf{Hypothesis A2 - Second Melnikov condition in measure estimates.}
 There exist   positive constants $\alpha_1,\alpha_2$ and $c_3$ such that the following holds: for each $0<\kappa<1/4$ and $K>0$ there exists a closed
 subset $ D':= D'(\kappa,K)\subset  D_0$ with
 ${\rm Meas}( D_0\setminus  D')\leq c_3K^{\alpha_1}\kappa^{\alpha_2}$ such that for all $\omega\in  D',$  $k\in \Z^n$ with $0<|k|\leq K$ and  $a,b\in \Z_{+}$ we have
 $|\la k,\omega\ra +\lambda_a-\lambda_b|\geq \kappa(1+|a-b|).$ Then  we have the following reducibility results.\\
\begin{Theorem}\label{KAM}
Given a non autonomous Hamiltonian (\ref{hameq}), we assume  that $(\lambda_{a})_{a\in \Z_{+}}$ satisfies Hypothesis  A1-A2 and  $P(\theta)\in \mathcal{C}^\beta(\T^n,\mathcal{M}_{\alpha})$ with $\alpha>0$ and $\beta>\max\{9(2+\frac{1}{\alpha})\frac{\gamma_1}{{\gamma_2}-24\delta},\ 9n,\ 24\}$ where  $\gamma_1= \max\{\alpha_1, n+3\},\
\gamma_2=\frac{\alpha\alpha_2}{5+2\alpha\alpha_2}$, $\delta\in(0,\frac{\gamma_2}{24})$.\\
\indent  Then there exists $\varepsilon_*(n,\beta,\delta )>0$ such that
if $0\leq \varepsilon<\varepsilon_*(n,\beta,\delta )$, there exist\\
(i) a Cantor set $D_\varepsilon\subset D_0$ with ${\rm Meas}(D_0\setminus D_\varepsilon)\leq c(n,\beta,\delta)\varepsilon^{\frac{3\delta}{2+\alpha^{-1}}}$;\\
(ii) a $\mathcal{C}^1$ family in $\omega\in D_\varepsilon$(in Whitney sense), linear, unitary and symplectic
coordinate transformation
$\Phi_\omega^\infty(\theta): Y_0\rightarrow Y_0,\ \theta\in\T^n,\ \omega\in D_\varepsilon,$ of the form
\begin{equation*}\label{transf}
(\xi_+,\eta_+)\mapsto(\xi,\eta)=\Phi_\omega^\infty(\theta)(\xi_+,\eta_+)=( \overline{M}_\omega(\theta)\xi_+,{M}_\omega(\theta)\eta_+),
\end{equation*}
where
 $\Phi_{\omega}^{\infty}(\theta)-id \in \mathcal{C}^{\iota}(\T^n, \mathcal{L}(Y_{0}, Y_{2\alpha})) \cap \mathcal{C}^{\iota}(\T^n, \mathcal{L}(Y_{s'}, Y_{s'}))$ with $ 0\leq s'\leq 1, \iota\leq\frac{2}{9}\beta$, $\iota\notin\Z$  and
 satisfies
\begin{eqnarray*}
\|\Phi_\omega^\infty-id\|_{\mathcal{C}^{\iota}(\T^n, \mathfrak{L}(Y_{0}, Y_{2\alpha})) }&\leq&   C(n,\beta, \iota) \varepsilon^{\frac{3}{2\beta}(\frac{2}{9}\beta-\iota)},
\end{eqnarray*}
and
\begin{eqnarray*}
\|\Phi_\omega^\infty-id\|_{\mathcal{C}^{\iota}(\T^n, \mathfrak{L}(Y_{s'}, Y_{s'})) }&\leq&   C(n,\beta, \iota) \varepsilon^{\frac{3}{2\beta}(\frac{2}{9}\beta-\iota)}.
\end{eqnarray*}
(iii)  a $\mathcal{C}^1$ family of  autonomous quadratic Hamiltonians in  normal forms
$$H_\infty(\xi_+,\eta_+)=\la\xi_+,N_\infty(\omega)\eta_+\ra=\sum\limits_{j\geq 1}\lambda_{j}^{\infty}\xi_{j,+}\eta_{j,+},\ \omega\in D_\varepsilon,$$
where $N_\infty(\omega)=diag\{\lambda_j^{\infty}\}$ is diagonal and is close to
$N$, i.e. 
\begin{eqnarray}\label{Ninfty}
[N_\infty(\omega)-N]_{\alpha}^{D_\varepsilon}\leq  c(n,\beta)\varepsilon, 
\end{eqnarray}
such that
$$H(t,\Phi_\omega^\infty(\omega t)(\xi_+,\eta_+))=H_\infty(\xi_+,\eta_+),\ t\in\R,\ (\xi_+,\eta_+)\in Y_{1},\ \omega\in D_\varepsilon. $$
Furthermore $\Phi^{\infty}_{\omega}(\theta)$ and $\Phi^{\infty}_{\omega}(\theta)^{-1}$ are bounded operators from
$Y_{s'}$ into itself for $0\leq s'\leq 1$ and they satisfy:
$$\|M_{\omega}(\theta)-Id\|_{\mathcal{L}(\ell_{s'}^2, \ell_{s'}^2)}, \|M^{-1}_{\omega}(\theta)-Id\|_{\mathcal{L}(\ell_{s'}^2, \ell_{s'}^2)}\leq  c \varepsilon^{1/2}. $$
\end{Theorem}

\section{Application to the Quantum Harmonic Oscillator--Proof of  Main Theorems}\label{s4}
\noindent In this section we will apply Theorem \ref{KAM} to the equation (\ref{maineq}) to prove Theorem \ref{quantumth}. For readers' convenience, we rewrite the equation
\begin{eqnarray}\label{maineq1sub1}
{\rm i}\partial_t u =-\partial_x^2 u+x^2 u+\varepsilon \la x\ra^{\mu}W(\nu x, \omega t)u,\ \ u=u(t,x),\ x\in\R,
 \end{eqnarray}
 where $0\leq \mu<\frac13$ and the potential $W(\varphi,\theta): \T^d\times  \T^n \mapsto
\R$ satisfies all the conditions in Theorem \ref{quantumth}. Following \cite{EK0}, we  expand $u$ and $\overline{u}$ on the Hermite basis $\{h_j\}_{j\geq1}$, namely, 
$u=\sum_{j\geq 1}\xi_jh_j$ and $ \overline{u}=\sum_{j\geq 1}\eta_j\overline{h}_j.$
And thus  (\ref{maineq1sub1}) can be written as a nonautonomous Hamiltonian system
\begin{eqnarray}\label{hs01}
\left\{
\begin{array}{cc}
\displaystyle\dot{\xi}_j=-{\rm i}\frac{\partial H}{\partial \eta_j}=-{\rm i}(2j-1)\xi_j-{\rm i}\varepsilon\frac{\partial}{\partial \eta_j}p(t,\xi,\eta),\ j\geq 1,\\
\displaystyle\dot{\eta}_j=\ \ {\rm i}\frac{\partial H}{\partial \xi_j}= \ \ {\rm i}(2j-1)\eta_j+{\rm i}\varepsilon\frac{\partial}{\partial {\xi}_j}p(t,\xi,\eta),\   j\geq 1,
\end{array}
\right.
\end{eqnarray}
where
\begin{equation}\label{hameqbeginning4}
H(t,\xi,\eta)=n(\omega)+p(t,\xi,\eta)= \la\xi, N\eta\ra+ \varepsilon\la\xi, P(\omega t)\eta\ra, \quad (\xi,\eta)\in Y_1\subset Y_0,
\end{equation}
and 
$
n(\omega):= \sum_{j\geq 1}
(2j-1)\xi_j\eta_j
$
and $P_{i}^j(\omega t)= \int_{\R} \la x\ra^{\mu}W(\nu x, \theta)h_i(x)\overline{h_j(x)}dx$.
 Here the external  parameters are the frequencies $\omega=(\omega_j)_{1\leq j\leq n}\in D_0:=[1, 2]^n$.
The proofs for the following two lemmas  are standard.
\begin{lemma}\label{aspt}
When $\lambda_a=2a-1,\ a\in \Z_{+},$ Hypothesis $\mathrm{A1}$ holds true with $c_0=c_2=1$ and $c_1=2$.
\end{lemma}
\begin{lemma}\label{aspt02}
When $\lambda_a=2a-1,\ a\in\Z_{+},$ Hypothesis $\mathrm{A2}$ holds true with $D_0=[0,2\pi]^n$, $ \alpha_1=n+1,\ \alpha_2=1,\  \ c_3= c(n )$ and
$$\ D':=\{\omega\in[0,2\pi]^n\big|\ |\la k,\omega\ra+j|\geq\kappa(1+|j|), {\rm\ for\ all\ }j{\rm\ \in\Z\ and\ } k{\rm\in\Z^n\setminus\{0\}}\}.$$
\end{lemma}

For the following we define the set $R_{\gamma, k}^{\tau}=\{\nu\in \R^d: |\la k, \nu\ra|<\frac{\gamma}{|k|^{\tau}}\}$ for $k\neq 0$ and
$R_{\gamma}^{\tau}=\bigcup\limits_{0\neq k\in Z^d}R^{\tau}_{\gamma, k}$. It is well known that
$Meas(R^{\tau}_{\gamma,k}\cap [A, B]^d)=O(\gamma/|k|^{\tau+1}),$
and thus
$Meas(R_{\gamma}^{\tau}\cap [A, B]^d)\leq O(\gamma)$.
We define the set
$\Omega_\gamma  : =[A, B]^d\setminus R_{\gamma}^{\tau}$
and then
$Meas([A,B]^d\setminus \Omega_\gamma)=\mathcal{O}(\gamma)$ as $\gamma\rightarrow 0$ and $\tau>d-1$.
\begin{lemma}\label{L3.3}
If $W(\varphi, \theta)\in C^{s}(\T^d\times \T^n)$ with $s\geq d+[1\vee \tau]+n+\beta_1+3$ and $\nu\in \Omega_{\gamma}$ and $\beta_1\in \N$, then
there exists $\alpha>0$ such that
the matrix function $P(\theta)$ defined by
$$(P(\theta))_{i}^{j}=\int_{\R}\la x\ra^{\mu} W(\nu x, \theta)h_i(x)\overline{h_j(x)}dx,\qquad i,j\geq 1, $$
belongs to ${\mathcal{C}^{\beta_1}}(\T^n, \mathcal{M}_{\alpha})$ with $\alpha=\frac{1}{12}-\frac \mu4$.
\end{lemma}
\begin{proof}We divide the proof into several steps.\\
\indent(a) We show that $P(\theta)\in\mathcal{M}_{\alpha}$. Since
$W(\nu x,\theta)=\sum\limits_{\substack{k\in\Z^d,\\l\in\Z^n}}\widehat{W}(k,l)e^{{\rm i}k\cdot\nu x}e^{{\rm i}l\theta},$
then
\begin{align*}
\big(P(\theta)\big)_i^j&=\int_\R \sum_{\substack{k\in\Z^d,\\l\in\Z^n}}\widehat{W}(k,l)e^{{\rm i}k\cdot\nu x}e^{{\rm i}l\theta}\la x\ra^{\mu}h_i(x)\overline{h_j(x)}dx\\
&=\sum_{l\in\Z^n}e^{{\rm i}l\theta}\sum_{k\ne 0}\widehat{W}(k,l)\int_\R e^{{\rm i}k\cdot\nu x}\la x\ra^{\mu}h_i(x)\overline{h_j(x)}dx,
\end{align*}
where we use (\ref{symmetry}). Note $\nu\in \Omega_{\gamma}$, we have
$|k\cdot\nu|\ge \frac{\gamma}{|k|^\tau}$ for $\tau>d-1$.
Thus by Lemma \ref{Indecaysection1}
\begin{align}
|\big(P(\theta)\big)_i^j|&\leq\frac{C(\gamma)}{i^\alpha j^\alpha}\sum_{l\in\Z^n}\sum_{k\ne 0}|\widehat{W}(k,l)|\cdot|k|^{1\vee\tau}.   \label{wkl1}
\end{align}
Denote $s_1=[1\vee\tau]+d+1$ and $s_2=n+1$. When $W(\varphi,\theta)\in C^{s_1+s_2}(\T^d\times \T^n)$, we have
$$|\widehat{W}(k,l)|\le \frac{d^{s_1}n^{s_2}}{|k|^{s_1}|l|^{s_2}}\sup_{\substack{|\alpha|=s_1\\|\nu|=s_2}} |\partial_\varphi^\alpha\partial_\theta^\nu W(\varphi,\theta)|,\quad \forall k\ne0,l\ne0,$$
and
$$
|\widehat{W}(k,0)|\le\frac{d^{s_1}}{|k|^{s_1}}\sup_{\substack{|\alpha|\le s_1\\|\nu|\le s_2}} |\partial_{\varphi}^{\alpha}\partial_{\theta}^\nu W(\varphi,\theta)|.
$$
From the choice of $s_1$ and $s_2$ and a straightforward computation
we have
\begin{eqnarray}\label{wkl2}
\sum_{l\in\Z^n}\sum_{k\ne 0}|\widehat{W}(k,l)|\cdot|k|^{1\vee\tau}\le C(d,\tau,n)\cdot\sup_{\substack{|\alpha|\le s_1\\|\nu| \le s_2}}|\partial_{\varphi}^{\alpha}\partial_{\theta}^\nu W(\varphi,\theta)|.
\end{eqnarray}
From (\ref{wkl1}) and (\ref{wkl2}),
$|P(\theta)|_{\alpha}\le C(\gamma,d,\tau,n)\cdot\sup\limits_{\substack{|\alpha|\le[1\vee\tau]+d+1\\|\nu| \le n+1}}|\partial_{\varphi}^{\alpha}\partial_{\theta}^\nu W(\varphi,\theta)|$
which follows $P(\theta)\in\mathcal{M}_{\alpha}$.\\
(b)\quad We show that $P(\theta)\in C^0(\T^n,\mathcal{M}_\alpha)$. For $\forall \theta_1,\theta_2\in\R^n,i,j\ge1$,
\begin{align*}
\Big|(P(\theta_1)-P(\theta_2))_i^j|&=|\int_\R \la x\ra^\mu (W(\nu x,\theta_1)-W(\nu x,\theta_2))h_i(x)\overline{h_j(x)}dx\Big|\\
&=\Big|\sum_{l\in\Z^n}(e^{{\rm i}l\theta_1}-e^{{\rm i}l\theta_2})\big(\sum_{k\ne0}\widehat{W}\left(k,l\right)\int_{\R}e^{{\rm i}k\cdot\nu x}\la x\ra^\mu h_i(x)\overline{h_j(x)}dx\big)\Big|\\
&\le\frac{C(\gamma,d,\tau,n)}{i^\alpha j^\alpha}\|\theta_1-\theta_2\| \sup_{\substack{|\alpha|\le s_1\\ |\nu|\le s_2+1}}|\partial_{\varphi}^{\alpha}\partial_{\theta}^\nu W(\varphi,\theta)|.
\end{align*}
Therefore $|P(\theta_1)-P(\theta_2)|_{\alpha}\to0$ when $\|\theta_1-\theta_2\|\to0$. Thus, $P(\theta)\in C^0(\T^n,\mathcal{M}_\alpha)$.\\
(c)\quad We show that $P(\theta)$ is Fr\'echet differentiable at each $\theta\in\T^n$. In fact, for any given $\theta_0\in \R^n$ we will prove that
$P^\prime(\theta_0)\in\mathfrak{L}(\R^n,\mathcal{M}_\alpha)$,
and for $\forall\xi\in\R^n$, $i,j\ge1$,
$$(P^\prime(\theta_0)\xi)_i^j=\int_{\R}\la x\ra^\mu \la W_\theta(\nu x,\theta_0),\xi\ra h_i(x)\overline{h_j(x)}dx.$$
We first define the right term by
\begin{align*}
(\mathcal{A}\xi)_{i}^{j}&:=\int_{\R}\la x\ra^\mu \la W_\theta(\nu x,\theta_0),\xi\ra h_i(x)\overline{h_j(x)}dx\\
&=\sum_{t=1}^{n}\xi_t\int_\R\la x\ra^\mu W_{\theta_t}(\nu x,\theta_0) h_i(x)\overline{h_j(x)}dx.
\end{align*}
Clearly,  $\mathcal{A}$ is a linear map on $\R^n$. From a similar computation we have 
\begin{align*}
 \|\mathcal{A}\|_{\mathfrak{L}(\R^n, \mathcal{M}_\alpha)}\leq C(\gamma,d,\tau,n)\sup_{\substack{|\alpha|\le s_1,\\ |\nu|\le s_2+1}}|\partial_{\varphi}^{\alpha}\partial_{\theta}^\nu W(\varphi,\theta)|.
\end{align*}
Similarly,  one obtains
\begin{align*}
|P(\theta)-P(\theta_0)-\mathcal{A}(\theta-\theta_0)|_\alpha &\le C(\gamma,n)\|\theta-\theta_0\|^2 \sum_{\substack{k\ne0,k\in\Z^d\\l\in\Z^n}}|\widehat{W}(k,l)||k|^{1\vee\tau}|l|^2\\
&\le C(\gamma,d,\tau,n)\|\theta-\theta_0\|^2 \sup_{\substack{|\alpha|\le[1\vee\tau]+d+1,\\|\nu|\le n+3}}|\partial_\varphi^\alpha\partial_\theta^\nu W(\varphi,\theta)|,
\end{align*}
which means that $P(\theta)$ is Fr\'echet differentiable on $\theta_0\in\T^n$ and $P^\prime(\theta_0)=\mathcal{A}$.\\
(d)\quad By a straightforward computation we can show that  $P(\theta)\in \mathcal{C}^1(\T^n, \mathcal{M}_{\alpha})$ since 
$$ \|P^\prime\left(\theta_1\right)-P^\prime\left(\theta_2\right)\|_{\mathfrak{L}(\R^n, \mathcal{M}_{\alpha})}\le C(\gamma,d,\tau,n)\cdot\|\theta_1-\theta_2\| \sup_{\substack{|\alpha|\le[1\vee\tau]+d+1\\|\nu|\le n+3}}|\partial_\varphi^\alpha\partial_\theta^\nu W|.$$
(e)\quad Inductively, we assume that $P(\theta)\in \mathcal{C}^m(\T^n, \mathcal{M}_{\alpha}),$ $m\leq \beta_1-1$, with
$$\left(P^{(m)}(\theta) (\xi_1,\cdots,\xi_m)\right)_i^j=\int_{\R}\la x\ra^\mu W_{\theta}^{(m)}(\nu x,\theta) (\xi_1,\cdots,\xi_m)h_i(x)\overline{h_j(x)}dx$$ satisfying
\begin{align*}
 \|P^{(m)}(\theta)\|_{\mathfrak{L}_m(\R^n, \mathcal{M}_{\alpha})}\le C(\gamma,d,\tau,n)\sup_{\substack{|\alpha|\le[1\vee\tau]+d+1,\\|\nu|\le n+m+1}}|\partial_\varphi^\alpha\partial_\theta^\nu W|,
\end{align*}
where $\mathfrak{L}_m(\R^n, \mathcal{M}_{\alpha})$ denotes the multi-linear operator space $\mathfrak{L}(\underbrace{\R^n\times\cdots\times\R^n}_{m}, \mathcal{M}_{\alpha}).$
Then we show that
$P(\theta)\in \mathcal{C}^{m+1}(\T^n, \mathcal{M}_{\alpha})$  with
$$\left(P^{(m+1)}(\theta) (\xi_1,\cdots,\xi_{m+1})\right)_i^j=\int_{\R}\la x\ra^\mu W_{\theta}^{(m+1)}(\nu x,\theta) (\xi_1,\cdots,\xi_{m+1})h_i(x)\overline{h_j(x)}dx,$$
and
\begin{align*}
 \|P^{(m+1)}(\theta)\|_{\mathfrak{L}_{m+1}(\R^n, \mathcal{M}_{\alpha})}\le C(\gamma,d,\tau,n)\sup_{\substack{|\alpha|\le[1\vee\tau]+d+1,\\|\nu|\le n+m+2}}|\partial_\varphi^\alpha\partial_\theta^\nu W|,
\end{align*}
We follow the method  in steps (c) and (d), and divide the proof into two parts ($e_1$) and ($e_2$) respectively.\\
($e_1$)\quad We show that $P^{\left(m\right)}\left(\theta\right)$ is Fr\'echet differentiable and for $\forall \theta_0\in\R^n$, $i,j\in\Z$,
$$\left(P^{\left(m+1\right)}\left(\theta_0\right)(\xi_1,\cdots,\xi_{m+1})\right)_i^j=\int_{\R} \la x\ra^\mu W_{\theta}^{(m+1)}(\nu x,\theta_0) (\xi_1,\cdots,\xi_{m+1})h_i(x)\overline{h_j(x)}dx,$$
with
\begin{align*}
 \|P^{(m+1)}(\theta)\|_{\mathfrak{L}_{m+1}(\R^n, \mathcal{M}_{\alpha})}\le C(\gamma,d,\tau,n)\sup_{\substack{|\alpha|\le[1\vee\tau]+d+1,\\|\nu|\le n+m+2}}|\partial_\varphi^\alpha\partial_\theta^\nu W|,
\end{align*}
In fact, given $\theta_0\in\R^n$, we define for $\xi_1,\cdots,\xi_{m+1}\in\R^n,\ i,j\in \Z$,
$$\left(\mathcal{B}\left(\xi_1,\cdots,\xi_{m+1}\right)\right)_i^j:=\int_{\R}\la x\ra^\mu W_{\theta}^{(m+1)}(\nu x,\theta_0) (\xi_1,\cdots,\xi_{m+1})h_i(x)\overline{h_j(x)}dx.  $$
 Since
\begin{align*}
|\left(\mathcal{B}\left(\xi_1,\cdots,\xi_{m+1}\right)\right)_i^j|&\le C(n)\|\xi_1\|\cdots\|\xi_{m+1}\| \sum_{\substack{k\ne0,k\in\Z^d\\l\in\Z^d}} |\widehat{W}(k,l)||l|^{m+1}\int_\R\la x\ra^\mu e^{{\rm i}k\cdot\nu x} h_i(x)\overline{h_j(x)}dx\\
&\le\frac{C(\gamma,d,\tau,n)\|\xi_1\|\cdots\|\xi_{m+1}\|}{i^\alpha j^\alpha}\sup_{\substack{|\alpha|\le [1\vee\tau]+d+1,\\ |\nu|\le n+m+2}}|\partial_{\varphi}^{\alpha}\partial_{\theta}^{\nu} W(\varphi,\theta)|,
\end{align*}
it follows that
$
 \|\mathcal{B}\|_{\mathfrak{L}_{m+1}\left(\R^n, \mathcal{M}_{\alpha}\right)}\leq C(\gamma,d,\tau,n)\sup\limits_{\substack{|\alpha|\le [1\vee\tau]+d+1,\\ |\nu|\le n+m+2}}|\partial_{\varphi}^{\alpha}\partial_{\theta}^\nu W(\varphi,\theta)|.
$
By a similar computation,
\begin{align*}
&\|P^{\left(m\right)}\left(\theta\right)-P^{\left(m\right)}\left(\theta_0\right) -\mathcal{B}\left(\theta-\theta_0\right)\|_{\mathfrak{L}_m(\R^n,\mathcal{M}_\alpha)}\\
&\le C(\gamma,n)\cdot\sum_{\substack{k\in\Z^d,k\ne0\\l\in\Z^n}}|k|^{1\vee\tau} |\widehat{W}(k,l)||l|^{m+2}\cdot\|\theta-\theta_0\|^2\\
&\le C(\gamma,d,\tau,n)\sup_{\substack{|\alpha|\le[1\vee\tau]+d+1,\\ |\nu|\le n+m+3}}|\partial_\varphi^\alpha\partial_\theta^\nu W|\cdot \|\theta-\theta_0\|^2.
\end{align*}
Thus, $P^{\left(m\right)}\left(\theta\right)$ is Fr\'echet differentiable  and $P^{\left(m+1\right)}(\theta_0)=\mathcal{B}$.\\
 ($e_2$)\quad For simplicity, we denote
 $({\rm i}l)^{(m+1)}:=\underbrace{({\rm i}l)\otimes\cdots\otimes({\rm i}l)}_{m+1}\in\mathfrak{L}_{m+1}(\R^n)$. Note
\begin{align*}
&\int_{\R}\la x\ra^\mu W_\theta^{(m+1)}(\nu x,\theta)(\xi_1,\cdots,\xi_{m+1}) h_i(x)\overline{h_j(x)}dx\\
=& \sum_{\substack{k\ne0,k\in\Z^d\\l\in\Z^n}}\widehat{W}(k,l)e^{{\rm i}l\theta} ({\rm i}l)^{m+1}(\xi_1,\cdots, \xi_{m+1})\int_\R \la x\ra^\mu e^{{\rm i}k\cdot\nu x} h_i(x)\overline{h_j(x)}dx.
\end{align*}
therefore,
\begin{align*}
&\Big|\big(\left(P^{(m+1)}\left(\theta_1\right)-P^{(m+1)}\left(\theta_2\right)\right)(\xi_1\cdots\xi_{m+1})\big)_i^j\Big|\\
\le&\frac{C(\gamma,n)\|\xi_1\|\cdots\|\xi_{m+1}\|\cdot\|\theta_1-\theta_2\|}{i^\alpha j^\alpha}\sum_{\substack{k\ne0,k\in\Z^d\\l\in\Z^n}}|\widehat{W}(k,l)|\cdot|l|^{m+2}\cdot|k|^{1\vee\tau}\\
\le&\frac{C(\gamma,d,\tau,n)\|\xi_1\|\cdots\|\xi_{m+1}\|\cdot\|\theta_1-\theta_2\|}{i^\alpha j^\alpha}\cdot \sup_{\substack{|\alpha|\le[1\vee\tau]+d+1,\\|\nu|\le n+m+3}}|\partial_\varphi^\alpha\partial_\theta^\nu W|.
\end{align*}
It follows that
$$ \|P^{(m+1)}\left(\theta_1\right)-P^{(m+1)}\left(\theta_2\right)\|_{\mathfrak{L}_{m+1}(\R^n, \mathcal{M}_{\alpha})}\le C(\gamma,d,\tau,n)\cdot\|\theta_1-\theta_2\|\cdot \sup_{\substack{|\alpha|\le[1\vee\tau]+d+1,\\|\nu|\le n+m+3}}|\partial_\varphi^\alpha\partial_\theta^\nu W|.$$
which means that $ \|P^{(m+1)}\left(\theta_1\right)-P^{(m+1)}\left(\theta_2\right)\|_{\mathfrak{L}_{m+1}(\R^n, \mathcal{M}_{\alpha})}\rightarrow0$ as $\|\theta_1-\theta_2\| \rightarrow0$. Thus we finish the induction.
\end{proof}
\noindent Proof of Theorem \ref{quantumth}:   It is clear  that
the Schr\"odinger equation (\ref{maineq1sub1}) is equivalent to Hamiltonian system (\ref{hs01}) with $\lambda_a=2a-1$. By
 lemmas given above, we can apply Theorem \ref{KAM}  to
  (\ref{hs01})  with $\gamma_1= n+3,\
\gamma_2=\frac{\alpha }{5+2\alpha }$ and $\delta= \frac{\gamma_2}{48}$.  This leads to Theorem \ref{quantumth}.\\
 \indent More precisely, in the new coordinates given in Theorem \ref{KAM}, $(\xi,\eta)=(\overline{ {M}}_\omega\xi_+, {M}_\omega\eta_+)$,
 system (\ref{hs01}) becomes autonomous and the systems are changed into the following:
 \begin{eqnarray*}
\left\{\begin{array}{c}
\dot{\xi}_{+,a}=-\mathrm{i} \lambda_a^{\infty}(\omega)\xi_{+,{a}},\\
\dot{\eta}_{+,a}= \ \ \mathrm{i} \lambda_a^{\infty}(\omega)\eta_{+,{a}},
\end{array}\right.\ \ \ a\in \Z_{+}.\label{inftyeq}
\end{eqnarray*}
 Hence the solution starts from $(\xi_+(0),\eta_+(0))$ is given by  $$(\xi_+(t),\eta_+(t))=(e^{-\mathrm{i}{t{N}}_\infty }\xi_+(0),e^{ \mathrm{i}{t {N}}_\infty }\eta_+(0)),\ t\in\R,$$
 where $N_{\infty}=diag\{\lambda^{\infty}_a\}_{a\in \Z_{+}}$.
Then the solution $u(t,x)$ of (\ref{maineq}) corresponding to the initial data $u_0(x)=\sum\limits_{a\geq 1}\xi_a(0)h_a(x)\in \mathcal{H}^{1}$  is formulated by
$u(t,x)=\sum\limits_{a\geq 1}\xi_a(t)h_a(x)$ with
$
\xi(t)=\overline{M}_\omega(\omega t)e^{-\mathrm{i}{t\overline{N}}_\infty }M^{T}_\omega(0)\xi(0),
$
where we use the fact $(\overline{ {M}}_\omega)^{-1}=M^{T}_\omega.$\\
\indent Let us define the transformation $\Psi^{\infty,1}_\omega(\theta)$ by
$$\Psi^{\infty,1}_\omega(\theta)(\sum_{a\geq 1}\xi_ah_a(x)):=\sum_{a\geq 1}(M^T_\omega(\theta)\xi)_ah_a(x)=\sum_{a\geq 1}\xi_{+,a}h_a(x).$$
\noindent  $u(t,x)$ satisfies (\ref{maineq}) if and only if  $v(t,x)=\Psi^{\infty,1}_\omega(\omega t)u(t,x)$ satisfies the autonomous equation
$i\partial_t v= (-\partial_{xx}+|x|^2)v+\varepsilon Q_1v,$
where  $$\varepsilon Q_1(\sum\limits_{a\in\Z_{+}}\xi_{a}h_a(x))=\sum\limits_{a\in\Z_{+}}((N_\infty-N_0)\xi)_ah_a(x)=\sum\limits_{a\in \Z_{+}}(\lambda_a^{\infty}-\lambda_a)\xi_ah_a(x). $$
For the rest estimates see lemma \ref{psismooth} below and  (\ref{Ninfty}). \qed
\begin{lemma}\label{psismooth}
$$\|\Psi^{\infty,1}_\omega(\cdot)-id\|_{\mathcal{C}^{\iota}(\T^n, \mathfrak{L}(\mathcal{H}^{0},\mathcal{H}^{2\alpha}))}\leq C \varepsilon^{\frac{3}{2\beta}(\frac{2}{9}\beta-\iota)},$$
and
\begin{eqnarray*}\label{lemma3.5}
\|\Psi^{\infty,1}_\omega(\cdot)-id\|_{\mathcal{C}^{\iota}(\T^n, \mathfrak{L}(\mathcal{H}^{s'},\mathcal{H}^{s'}))}\leq C \varepsilon^{\frac{3}{2\beta}(\frac{2}{9}\beta-\iota)},
\end{eqnarray*}
where $\iota$ is defined in Theorem \ref{KAM} and $0\leq s'\leq 1$.
\end{lemma}
We delay the above proof in section \ref{S5}. \\

 \section{Estimates on eigenfunctions}\label{S2}
\subsection{Langer's turning point}
We now introduce Langer's turning point method based on the contents in Chapter 22.27 of \cite{T2}.  For other application of Langer's turning point theory, see \cite{WLiang17, Yajima}.  \\
\indent Consider the function
\begin{eqnarray}\label{tezhengfangcheng}
\psi^{\prime\prime}(x)+(\lambda-q(x))\psi(x)=0,\quad x>0,
\end{eqnarray}
where $q(x)$ increases steadily to $+\infty$, $q(x)$ is three times differentiable, and for $x>x_0$ for some positive constant $x_0$, $q'(x)$ is nondecreasing, and  as $x\rightarrow \infty$
$$\frac{q'(x)}{q(x)}=O(\frac1x),\quad \frac{q^{\prime\prime}(x)}{q'(x)}=O(\frac1x),\quad \frac{q^{\prime\prime\prime}(x)}{q^{\prime\prime}(x)}=O(\frac1x).$$
We also suppose that there exists a unique $X>0$ such that $\lambda=q(X)$. \\
\indent Then for constant $a\ge1$,
$\ln\frac{q(ax)}{q(x)}=\int_x^{ax}\frac{q'(t)}{q(t)}dt=O(\int_x^{ax}\frac{1}{t}dt)=O(\ln a),$
which means
$q(ax)=O(q(x))$,
and similarly for $q'(x)$ and $q^{\prime\prime}(x)$. Since
$q(x)=\int_0^x q'(t)dt+q(0)\le xq'(x)+q(0),$
it follows $q'(x)\geq C\frac{q(x)}{x}$,
when $x\geq x_1$ for some $x_1>0$.  Now set
$\eta(x)=(\lambda-q(x))^\frac14\psi(x)$ and $\zeta(x)=\int_X^x (\lambda-q(t))^\frac12 dt$, 
where
\begin{align*}
\arg\zeta(x)=\left\{
\begin{array}{ll}
&\frac12\pi \quad(x>X),\\
&-\pi \quad(x<X).
\end{array}\right.
\end{align*}
Then the equation (\ref{tezhengfangcheng}) is transformed into
$\frac{d^2\eta}{d\zeta^2}+\eta+ [\frac{q^{\prime\prime}(x)}{4(\lambda-q(x))^2}+\frac{5q^\prime(x)^2}{16(\lambda-q(x))^3}]\eta=0$
and this may be expressed as
\begin{eqnarray}\label{xinfangcheng}
\frac{d^2\eta}{d\zeta^2}+(1+\frac{5}{36\zeta^2})\eta=f(x)\eta,
\end{eqnarray}
where
$f(x)=\frac{5}{36\zeta^2}-\frac{q^{\prime\prime}(x)}{4(\lambda-q(x))^2}-\frac{5q^\prime(x)^2}{16(\lambda-q(x))^3}$.
As we know, Bessel equation
$\frac{d^2 G}{d\zeta^2}+(1+\frac{5}{36\zeta^2})G=0$
has two linearly independent solutions $(\frac{\pi\zeta}{2})^\frac12 J_\frac13(\zeta)$ and $(\frac{\pi\zeta}{2})^\frac12 H_\frac13^{(1)}(\zeta)$, where $J_\nu(x)$ and $H_\nu^{(1)}(x)$ are the first kind Bessel function and one of the third kind Bessel function, respectively. By the property of Bessel function that
$x(J_\nu(x)H_\nu^{(1)\prime}(x)-J_\nu^\prime(x)H_\nu^{(1)}(x))=\frac{2{\rm i}}{\pi}$, 
then (\ref{xinfangcheng}) is formally equivalent to the integral equation
\begin{align*}
\eta = (\frac{\pi\zeta}{2})^\frac12H_{\frac13}^{(1)}(\zeta)+\frac{\pi{\rm i}}{2}\int_x^\infty \bigg(H_\frac13^{(1)}(\zeta)J_\frac13(\theta)-J_\frac13(\zeta)H_\frac13^{(1)}(\theta)\bigg)\zeta^\frac12\theta^\frac12 f(t)(\lambda-q(t))^\frac12\eta(t)dt,
\end{align*}
where we write $\zeta=\zeta(x)$ and $\theta=\zeta(t)$ for convenience. Set
$$\alpha(x)=e^{-{\rm i}\zeta}(\frac{\pi\zeta}{2})^\frac12 H_\frac13^{(1)}(\zeta),\quad \beta(x)=e^{{\rm i}\zeta}(\frac{\pi\zeta}{2})^\frac12 J_\frac13(\zeta),\quad \chi(x)=e^{-{\rm i}\zeta}\eta(x),$$
then
$$\chi(x)=\alpha(x)+{\rm i}\int_x^\infty \bigg(\alpha(x)\beta(t)-e^{2{\rm i}(\theta-\zeta)}\beta(x)\alpha(t)\bigg) f(t)(\lambda-q(t))^\frac12\chi(t)dt.$$
Clearly, $\alpha(x)$, $\beta(x)$ are bounded, and
$\Im(\theta-\zeta)=\Im(\int_x^t(\lambda-q(u))^\frac12du)\geq0$.
To give the estimate of solution of (\ref{xinfangcheng}) or (\ref{tezhengfangcheng}), we first present two preparation lemmas and delay the proofs in the Appendix.
\begin{lemma}{\rm (\cite{T2})}\label{6.5}
For fixed $\lambda$, if $x>2X$, then
$\int_x^\infty|f(t)(\lambda-q(t))^\frac12|dt\le\frac{C}{x(q(x))^\frac12}$,
here $C$ is a constant independent of $x$ and $\lambda$.
\end{lemma}
\begin{lemma}{\rm (\cite{T2})}\label{2.2sub}
$\int_0^\infty |f(x)||\lambda-q(x)|^\frac12dx=O\bigg(\frac{1}{X\lambda^\frac12}\bigg), \quad \lambda\to\infty.$
\end{lemma}
From these two lemmas, we can prove that the iteration converges. In fact, if we denote
$\int_0^\infty |f(t)||\lambda-q(t)|^\frac12dt= M_0=O\left(\frac1{X\lambda^\frac12}\right)$, 
and
$\left|\alpha(x)\beta(t)-e^{2{\rm i}(\theta-\zeta)}\beta(x)\alpha(t)\right|\le M$
uniformly, then
$|\chi_0(x)|=|\alpha(x)|\le C, \quad |\chi_1(x)-\chi_0(x)|\le C M M_0$, 
and generally, if
$|\chi_n(x)-\chi_{n-1}(x)|\le CM^n M_0^n$, 
then
\begin{align*}
&|\chi_{n+1}(x)-\chi_n(x)|\\
=&\left|\int_x^\infty \bigg(\alpha(x)\beta(t)-e^{2{\rm i}(\theta-\zeta)}\beta(x)\alpha(t)\bigg) f(t)(\lambda-q(t))^\frac12(\chi_n(t)-\chi_{n-1}(t))dt\right|\\
\le& CM^{n+1}M_0^n\int_x^\infty \left|f(t)(\lambda-q(t))^\frac12\right|dt\le CM^{n+1}M_0^{n+1}.
\end{align*}
Thus,
\begin{align*}
|\chi_n(x)|&\le|\chi_0(x)|+|\chi_1(x)-\chi_0(x)|+\cdots+|\chi_n(x)-\chi_{n-1}(x)|\\
&\le C(1+MM_0+\cdots+M^n M_0^n)\le\frac{C}{1-MM_0}.
\end{align*}
If $\lambda$ is sufficiently large, then $MM_0<1$, and by the theorem of dominated convergence, when $n\to\infty$, $\chi_n(x)\to\chi(x)=\alpha(x)+O(\frac{1}{X\lambda^\frac12})$ uniformly w.r.t $x$, which means that $\chi(x)$ is bounded.\\
Next we show that
\begin{eqnarray}\label{biaodaforchi}
\chi(x)=\alpha(x)\left(1+O\left(\frac{1}{X\lambda^\frac12}\right)\right).
\end{eqnarray}
In fact, similar as Lemma \ref{Bessel}, if $\zeta(x)<-c_0$ or ${\rm i}\zeta(x)<-c_0$, where $c_0$ are arbitrary two positive constants, we can prove that $|\alpha(x)|>C$ and (\ref{biaodaforchi}) holds. While for $0<|\zeta(x)|\leq c_0$ we have
$|\beta(x)|\le C|\alpha(x)|$. 
Thus,
$$|\chi(x)-\alpha(x)|=\left|\int_x^\infty \bigg(\alpha(x)\beta(t)-e^{2{\rm i}(\theta-\zeta)}\beta(x)\alpha(t)\bigg) f(t)(\lambda-q(t))^\frac12\chi(t)dt\right|\le \frac{C|\alpha(x)|}{X\lambda^\frac12}.$$
Hence we have
\begin{lemma}\cite{T2}\label{biaoda}
When $\lambda>c_1>0$ large enough such that 
$$M\int_0^{\infty}|f(t)||\lambda-q(x)|^\frac12dt<1,$$ the solution of (\ref{tezhengfangcheng}) can be written as
$\psi(x)=(\lambda-q(x))^{-\frac14}(\frac{\pi\zeta}{2})^\frac12 H_\frac13^{(1)}(\zeta)(1+O(\frac{1}{X\lambda^\frac12}))$.
\end{lemma}
\noindent Since 
$M\int_x^\infty|f(t)||\lambda-q(t)|^\frac12dt\le\frac{MC}{x(q(x))^\frac12}$,
where $M,C$ are independent of $x$ and $\lambda$, 
it follows $M\int_{c_2}^\infty|f(t)||\lambda-q(t)|^\frac12dt<\frac12$ for some positive constant $c_2$.
Hence we have
\begin{lemma}\cite{T2}\label{xiaobiaoda}
For any fixed $\lambda$, when $x>\max\{2X,c_2\}$, the solution of (\ref{tezhengfangcheng}) can be written as
$\psi(x)=\psi_1(x)+\psi_2(x)$, 
where $\psi_1(x)=(\lambda-q(x))^{-\frac14}(\frac{\pi\zeta}{2})^\frac12 H_\frac13^{(1)}(\zeta)$, and $|\psi_2(x)|\le\frac{C}{x(q(x))^\frac12}|\psi_1(x)|$.
\end{lemma}
\begin{remark}\label{hnbiaoda}
In the application  $q(x)=x^2$, $\lambda_n=2n-1$ with $n\in \Z_{+}$. Then for $\lambda_n>c_1$, i.e. $n>\frac{c_1+1}{2}$,
$h_n(x)=\psi_1^{(n)}(x)+\psi_2^{(n)}(x)$, 
where $x>0$
and
$\psi_1^{(n)}(x)=(\lambda_n-x^2)^{-\frac14}(\frac{\pi\zeta_n}{2})^\frac12 H_\frac13^{(1)}(\zeta_n)$ and 
$\psi_2^{(n)}(x)=\psi_1^{(n)}(x) O(\frac{1}{\lambda_n})$. While for $\lambda_n\le\frac{c_2^2}{4}: = c_3$ and $x>c_2$, we have
$h_n(x)=\psi_1^{(n)}(x)+\psi_2^{(n)}(x)$,  
where  $|\psi_2^{(n)}(x)|\le\frac{C}{x^2}|\psi_1^{(n)}(x)|$. 
\end{remark}
\begin{remark}\label{4.6}
For the following we denote $m_0=\max\{\frac{c_1+1}{2},\frac{c_3+1}{2}\}$.
\end{remark}

\subsection{Proof of Lemma \ref{Indecaysection1}}
A well-known fact is that
$h_n(x)={(n! 2^n \pi^{\frac12})^{-\frac12}}{e^{-\frac12 x^2} H_n(x)}$ where $H_n(x)$ is the Hermite polynomial of degree $n$ and $h_n(x)$ is an even or odd function of $x$ according to whether $n$ is odd or even(\cite{T1}).
From the symmetry of $h_n(x)$, we only need to consider $\displaystyle\int_0^{+\infty} \la x\ra^{\mu} e^{{\rm i}kx}h_m(x)\overline{h_n(x)}dx$ for $1\leq m\leq n$. 
Rewrite
\begin{eqnarray}\label{int0infty}
\int_0^{+\infty} \la x\ra^{\mu} e^{{\rm i}kx}h_m(x)\overline{h_n(x)}dx=\int_0^{X_n}+\int_{X_n}^{+\infty}.
\end{eqnarray}
 From Remark \ref{hnbiaoda} and Remark \ref{4.6} and  $m>m_0$,
\begin{eqnarray}\label{gujidecay1}
h_m(x)&=&(\lambda_m-x^2)^{-\frac{1}{4}}(\frac{\pi\zeta_m}{2})^{\frac{1}{2}}H_{\frac{1}{3}}^{(1)}(\zeta_m) +(\lambda_m-x^2)^{-\frac{1}{4}}(\frac{\pi\zeta_m}{2})^{\frac{1}{2}}H_{\frac{1}{3}}^{(1)}(\zeta_m)O (\frac{1}{\lambda_m}) \nonumber\\
&:=& \psi^{(m)}_1(x) +\psi^{(m)}_2(x),
\end{eqnarray}
where $\zeta_m(x)=\displaystyle\int_{X_{m}}^x\sqrt{\lambda_m-t^2}dt$ with $X_m^2=\lambda_m(X_m>0)$.
While for $m\le m_0$, by Lemma \ref{xiaobiaoda} and $x>2X_{m_0}$, 
$h_m(x)=\psi_1^{(m)}(x)+\psi_2^{(m)}(x)$,
where $\psi_1^{(m)}(x)=(\lambda_m-x^2)^{-\frac14}(\frac{\pi\zeta_m}{2})^\frac12 H_\frac13^{(1)}(\zeta_m)$, and $|\psi_2^{(m)}(x)|\le\frac{C}{x^2}|\psi_1^{(m)}(x)|$. Now we estimate (\ref{int0infty}) in the following three cases: \\
1)$m,n<C_*: = 2^8m_0^3$; 2)$m\le m_0$ and $n\ge C_*$; 3)$m,n>m_0$.
\begin{lemma}
When $n,m<C_*$,
$$\left|\int_0^{+\infty} \la x\ra^{\mu} e^{{\rm i}kx}h_m(x)\overline{h_n(x)}dx\right|\le\frac{C}{n^{\frac{1}{12}-\frac\mu4}m^{\frac{1}{12}-\frac\mu4}}.$$
\end{lemma}
\begin{proof}
Since $m<C_*$, then for $x>X_0$, $h_m(x)=\psi_1^{(m)}(x)+\psi_2^{(m)}(x)$,
where $\psi_1^{(m)}(x)=(\lambda_m-x^2)^{-\frac14}(\frac{\pi\zeta_m}{2})^\frac12 H_\frac13^{(1)}(\zeta_m)$, and $|\psi_2^{(m)}(x)|\le\frac{C}{x^2}|\psi_1^{(m)}(x)|$ and $X_0$ is a positive constant depending on $C_*$ only. $h_n(x)$ has a similar decomposition. \\
When $x\le X_0$, by $\rm{H\ddot{o}lder}$ inequality and $n,m<C_*$,
$$\left|\int_0^{X_0} \la x\ra^{\mu} e^{{\rm i}kx}h_m(x)\overline{h_n(x)}dx\right|\le X_0^\mu\le \frac{C}{n^{\frac{1}{12}-\frac\mu4}m^{\frac{1}{12}-\frac\mu4}}.$$
When $x>X_0$, $|X_m^2-x^2|^{-\frac14}<1$, and by Lemma \ref{Bessel}, $\left|\sqrt{\frac{\pi \zeta_m}{2}}H^{(1)}_\frac{1}{3}(\zeta_m)\right|\leq e^{-\left| \zeta_m\right|}$. By Lemma \ref{zetaesti} and  $x> X_0$,
$\left|\zeta_m\right| \geq\frac{2\sqrt2}{3}X_m^\frac12(x-X_m)^\frac32\geq x-X_0$. 
 Thus,
\begin{align*}
   \left|\int_{X_0}^{+\infty} \la x\ra^{\mu} e^{{\rm i}kx}h_m(x)\overline{h_n(x)}dx\right|
\le \int_{X_0}^{+\infty} \la x\ra^{\mu}e^{-2(x-X_0)}dx\le Ce^{2X_0}\le\frac{C}{n^{\frac{1}{12}-\frac\mu4}m^{\frac{1}{12}-\frac\mu4}}.
\end{align*}
\end{proof}

\begin{lemma}
For $m\le m_0$ and $n\ge C_*$,
$$\left|\int_0^{+\infty} \la x\ra^{\mu} e^{{\rm i}kx}h_m(x)\overline{h_n(x)}dx\right|\le\frac{C}{n^{\frac{1}{12}-\frac\mu4}m^{\frac{1}{12}-\frac\mu4}}.$$
\end{lemma}
\begin{proof}
We split the integral into  two parts
$$\int_0^{+\infty} \la x\ra^{\mu} e^{{\rm i}kx}h_m(x)\overline{h_n(x)}dx=\int_0^{X_n^\frac13}+\int_{X_n^\frac13}^{+\infty}.$$
 When $x>2X_{m_0}$, by Lemma \ref{xiaobiaoda} we have $|h_m(x)|\le 2(x^2-X_m^2)^{-\frac14}|\sqrt{\frac{\pi \zeta_m}{2}}H^{(1)}_\frac{1}{3}(\zeta_m)|\le 2 e^{-|\zeta_m|}$. On the other hand by Lemma \ref{Bessel}, $|h_n(x)|\le C(X_n^2-x^2)^{-\frac14}$ on $[0,X_n^\frac13]$. Thus,
$$\left|\int_0^{X_n^\frac13}\la x\ra^{\mu} e^{{\rm i}kx}h_m(x)\overline{h_n(x)}dx\right|\le C\int_0^{X_n^\frac13}\la x\ra^{\mu}(X_n^2-x^2)^{-\frac14}dx \le CX_n^{-\frac16+\frac\mu3}\le\frac{C}{n^{\frac{1}{12}-\frac\mu4}m^{\frac{1}{12}-\frac\mu4}}.$$
When $x\ge X_n^\frac13\ge2X_{m_0}$, by Lemma \ref{zetaesti}, $e^{-|\zeta_m|}\le e^{-C(x-X_m)}$. Thus, by $\rm{H\ddot{o}lder}$ inequality,
$$\left|\int_{X_n^\frac13}^{+\infty}\la x\ra^{\mu} e^{{\rm i}kx}h_m(x)\overline{h_n(x)}dx\right|\le C\left(\int_{X_n^\frac13}^{+\infty}\la x\ra^{2\mu}e^{-Cx}dx\right)^\frac12\le e^{-CX_n^\frac13}.$$
\end{proof}

Now we turn to the third case that is $m,n>m_0$. Rewrite
$\int_0^{+\infty} \la x\ra^{\mu} e^{{\rm i}kx}h_m(x)\overline{h_n(x)}dx=\int_0^{X_n}+\int_{X_n}^{+\infty}$.
We first turn to the integral $\int_{X_n}^{+\infty}$. In the following part of  this section we will denote 
$\cal{F}(x) = \la x\ra^\mu e^{{\rm i}kx}\psi_1^{(m)}(x) \overline{\psi_1^{(n)}(x)}$ for simplicity. 
\subsection{the integral on $[X_n, +\infty)$}
In fact in this part we have
\begin{lemma}\label{Xntowuqiong}
When $m_0<m\le n$,
$\displaystyle\left|\int_{X_n}^{+\infty}\la x\ra^\mu e^{{\rm i}kx}h_m(x)\overline{h_n(x)}dx\right|\leq\frac{C}{m^{\frac{1}{12}-\frac\mu4} n^{\frac{1}{12}-\frac\mu4}}. $
\end{lemma}
Lemma \ref{Xntowuqiong} results from the following two lemmas.
\begin{lemma}\label{2.2}
For $m_0<m\le n$, 
$$\displaystyle\left|\int_{2X_n}^{+\infty}\la x\ra^\mu e^{{\rm i}kx}h_m(x) \overline{h_n(x)}dx\right|\leq e^{-C n}.$$
\end{lemma}
\begin{proof}
Since $h_n(x):=\psi^{(n)}_1(x) +\psi^{(n)}_2(x)$, we only need to prove  
$\displaystyle\left|\int_{2X_n}^{+\infty}\cal{F}(x)dx\right|\leq e^{-C n}$, 
since the other three integrals have better estimates,  where $\psi^{(n)}_2(x)=O(\frac1{\lambda_n})\psi^{(n)}_1(x)$. 
 Recall that $\psi_1^{(n)}(x)=\left(\lambda_n - x^2 \right)^{-\frac{1}{4}}\sqrt{\frac{\pi \zeta_n}{2}}H^{(1)}_\frac{1}{3}(\zeta_n)$, then
by Lemma \ref{Bessel},  $\left|\sqrt{\frac{\pi \zeta_n}{2}}H^{(1)}_\frac{1}{3}(\zeta_n)\right|\leq e^{-\left| \zeta_n\right|}$ when $x\geq 2X_n$.  By Lemma \ref{zetaesti}, when $x\geq2X_n$,
$\left|\zeta_n\right| \geq \frac{2\sqrt2}{3}X_n(x-X_n)\geq\frac{\sqrt2}{3}(x-X_n)+\frac{\sqrt2}{3}X_n^2$. 
Therefore,
\begin{align*}
          \left|\int_{2X_n}^{+\infty}\cal{F}(x) dx\right| \leq &\int_{2X_n}^{+\infty}\la x\ra^\mu(x^2-\lambda_m)^{-\frac{1}{4}}(x^2-\lambda_n)^{-\frac{1}{4}} e^{-\left| \zeta_m\right|}e^{-\left| \zeta_n\right|}dx \\
\leq & \int_{2X_n}^{+\infty}\la x\ra^\mu(x^2-\lambda_n)^{-\frac{1}{4}}(x^2-\lambda_n)^{-\frac{1}{4}} e^{-\left| \zeta_n\right|}dx\\
\leq &C e^{-\frac{\sqrt2}{3}X_n^2}n^{-\frac{1}{2}} \int_{2X_n}^{+\infty}\la x\ra^\mu e^{-\frac{\sqrt2}{3}(x-X_n)}dx \leq e^{-C n}.
\end{align*}
\end{proof}

\begin{lemma}\label{2.3}
For  $m_0<m\le n$, 
$$\displaystyle\left|\int^{2X_n}_{X_n}\la x\ra^\mu e^{{\rm i}kx}h_m(x) \overline{h_n(x)}dx\right|\leq  \frac{C}{m^{\frac{1}{12}-\frac\mu4} n^{\frac{1}{12}-\frac\mu4}}.$$
\end{lemma}
\begin{proof}
As above we only need to estimate 
$|I|=\displaystyle\left|\int^{2X_n}_{X_n}\cal{F}(x) dx\right|\leq \frac{C}{m^{\frac{1}{12}-\frac\mu4} n^{\frac{1}{12}-\frac\mu4}}$.
We divide $I$ into two parts as
\begin{align*}
&|I|=\displaystyle\left|(\int^{2X_n}_{X_n+X_n^{\frac{1}{3}}}+\int^{X_n+X_n^{\frac{1}{3}}}_{X_n})\cal{F}(x) dx\right|\\
\leq& CX_n^\mu \Big(\int^{2X_n}_{X_n+X_n^{\frac{1}{3}}}+\int^{X_n+X_n^{\frac{1}{3}}}_{X_n}\Big)\left|\psi_1^{(m)}(x) \overline{\psi_1^{(n)}(x)}\right|dx.
\end{align*}
By Lemma \ref{zetaesti}, when $x\geq X_n+X_n^{\frac{1}{3}}$,
$\left|\zeta_n\right| \geq\frac{2\sqrt2}{3}X_n^\frac{1}{2}(x-X_n)^\frac{3}{2}\geq\frac{2\sqrt2}{3}X_n.$
Thus,
\begin{align*}
          \int^{2X_n}_{X_n+X_n^{\frac{1}{3}}}\left|\psi_1^{(m)}(x) \overline{\psi_1^{(n)}(x)}\right| dx
\leq & C\int^{2X_n}_{X_n+X_n^{\frac{1}{3}}}(x^2-\lambda_m)^{-\frac{1}{4}}(x^2-\lambda_n)^{-\frac{1}{4}} e^{-\left| \zeta_n\right|}dx \\
\leq & Ce^{-\frac{2\sqrt2}{3}X_n} \int^{2X_n}_{X_n+X_n^{\frac{1}{3}}}(x^2-\lambda_n)^{-\frac{1}{2}}dx\le  Ce^{-\frac{2\sqrt2}{3}X_n}.
\end{align*}
On the other hand,
\begin{align*}
          &\int^{X_n+X_n^{\frac{1}{3}}}_{X_n}\left|\psi_1^{(m)}(x) \overline{\psi_1^{(n)}(x)}\right|dx \leq C\int^{X_n+X_n^{\frac{1}{3}}}_{X_n}(x^2-\lambda_m)^{-\frac{1}{4}}(x^2-\lambda_n)^{-\frac{1}{4}}dx \\
          &\leq C\int^{X_n+X_n^{\frac{1}{3}}}_{X_n} (x^2-\lambda_n)^{-\frac{1}{2}}dx  \leq CX_n^{-\frac{1}{2}} \int^{X_n+X_n^{\frac{1}{3}}}_{X_n}(x-X_n)^{-\frac{1}{2}}dx \leq CX_n^{-\frac{1}{3}}.
\end{align*}
Therefore,  
$|I|\leq CX_n^{\mu-\frac13} \leq\frac{C}{m^{\frac{1}{12}-\frac\mu4}n^{\frac{1}{12}-\frac\mu4}}$. 
\end{proof}
\indent In  the following we will estimate the integral on $[0,X_n]$, for which we have to  discuss  two different cases, namely,  $X_n\geq 2X_m$ or $X_m\leq X_n \leq 2X_m$ with $n\geq m>m_0$.
\subsection{the integral on $[0, X_n]$ for the case $X_n\geq 2X_m$}
To simplify the following proof we will use the following notation in the remained parts.
We define
$
f_m(x)=\int_0^{\infty}e^{-t}t^{-\frac{1}{6}}\lrp{1+\frac{{\rm i}t}{2\zeta_m}}^{-\frac{1}{6}}dt
$
and
$
f_n(x)=\int_0^{\infty}e^{-t}t^{-\frac{1}{6}}\lrp{1+\frac{{\rm i}t}{2\zeta_n}}^{-\frac{1}{6}}dt.
$
 When $x\in [0,X_{m}]$,  from a straightforward computation
\begin{align*}
\psi_1^{(m)}(x) &= (X_m^2-x^2)^{-\frac{1}{4}}\sqrt{\frac{\pi\zeta_m}{2}} H_\frac{1}{3}^{(1)}(\zeta_m)\\
&= (X_m^2-x^2)^{-\frac{1}{4}} \frac{e^{{\rm i}\lrp{\zeta_m-\frac{\pi}{6}-\frac{\pi}{4}}}}{\Gamma{\lrp{\frac{5}{6}}}}
\int_0^{\infty}e^{-t}t^{-\frac{1}{6}}\lrp{1+\frac{{\rm i}t}{2\zeta_m}}^{-\frac{1}{6}}dt\\
&=C(X_m^2-x^2)^{-\frac{1}{4}}e^{{\rm i}\zeta_m(x)}f_m(x).
\end{align*}
Similarly, when $x\in [0,X_m]$,
$
\overline{\psi_1^{(n)}(x)} =C(X_n^2-x^2)^{-\frac{1}{4}}e^{-{\rm i}\zeta_n(x)}\overline{f_n(x)}.
$
We also define $\Psi(x)=(X_m^2-x^2)^{-\frac{1}{4}} (X_n^2-x^2)^{-\frac{1}{4}}\cdot f_m(x)\overline{f_n(x)}$
and
$g(x)=(\zeta_n(x)-\zeta_m(x)-kx)^\prime=\sqrt{X_n^2-x^2}-\sqrt{X_m^2-x^2}-k$ with  $x\in [0,X_m]$.
We will use the derivative of $\Psi$ for many times, i.e.
\begin{align*}
\Psi^\prime(x)=&\frac{1}{2}x(X_m^2-x^2)^{-\frac{5}{4}} (X_n^2-x^2)^{-\frac{1}{4}}\cdot f_m(x)\overline{f_n(x)}\\
+&\frac{1}{2}x(X_m^2-x^2)^{-\frac{1}{4}} (X_n^2-x^2)^{-\frac{5}{4}}\cdot f_m(x)\overline{f_n(x)}\\
+&(X_m^2-x^2)^{-\frac{1}{4}} (X_n^2-x^2)^{-\frac{1}{4}}\cdot \left(f_m^\prime(x)\overline{f_n(x)}+f_m(x)\overline{f_n^\prime(x)}\right).
\end{align*}
From $x\in [0, X_m]$ we obtain
 $\left|f_m(x)\right| \leq \Gamma(\frac{5}{6})$ and  $\left|f_n(x)\right| \leq \Gamma(\frac{5}{6})$. By a straightforward computation we have 
\begin{Corollary}\label{coro2.5}
For $x\in [0,X_m)$ and $m\leq  n$,
\begin{align*}
\left| \Psi^\prime(x)\right|\leq & C\Big(x(X_m^2-x^2)^{-\frac{5}{4}}(X_n^2-x^2)^{-\frac{1}{4}} +x(X_m^2-x^2)^{-\frac{1}{4}} (X_n^2-x^2)^{-\frac{5}{4}}+\\
& \frac{(X_m^2-x^2)^{\frac{1}{4}} (X_n^2-x^2)^{-\frac{1}{4}}}{X_m(X_m-x)^3} +
\frac{(X_m^2-x^2)^{-\frac{1}{4}} (X_n^2-x^2)^{\frac{1}{4}}}{X_n (X_n-x)^3}\Big)\\
=&C\big(J_1+J_2+J_3+J_4\big)\leq C(J_1+J_3).
\end{align*}
\end{Corollary}
Our main intent in this subsection is to set up
\begin{lemma}\label{lemma2.4}
For $k \neq 0$, if $X_n\geq 2X_m$, then $$\displaystyle\left|\int^{X_n}_0 \la x\ra^\mu e^{{\rm i}kx}h_m(x)\overline{h_n(x)}dx\right|\leq
\frac{C (|k|\vee 1)^{\frac{1}{2}}}{m^{\frac18-\frac\mu4} n^{\frac{1}{12}-\frac\mu4}},$$
where $m_0< m\leq n$.
\end{lemma}

We first have 
\begin{lemma}\label{lemma2.6}
For $k \neq 0$, if $X_n\geq2X_m$, then
$$\left|\int_0^{X_m-X_m^{-\frac{1}{3}}} \la x\ra^\mu e^{{\rm i}kx}h_m(x)\overline{h_n(x)}dx\right|\leq \frac{C(|k|\vee 1)^{\frac12}}{m^{\frac{1}{8}-\frac\mu4} n^{\frac{1}{8}-\frac\mu4}}, $$
where $m_0< m\leq n$.
\end{lemma}
\begin{proof}
First we estimate the main part
\begin{align*}
\int^{X_m-X_m^{-\frac{1}{3}}}_{0}\cal{F}(x) dx= C\int^{X_m-X_m^{-\frac{1}{3}}}_{0} \la x\ra^\mu e^{{\rm i}(\zeta_m-\zeta_n+kx)}\Psi(x)dx,\\
\end{align*}
by oscillatory integrals, where
$\Psi(x)=(X_m^2-x^2)^{-\frac{1}{4}} (X_n^2-x^2)^{-\frac{1}{4}}\cdot f_m(x)\overline{f_n(x)}.$
We discuss two different cases. \\
Case 1:  $k\leq\frac{X_n}{4}$. In this case, we have
\begin{align*}
g(x)\geq\sqrt{X_n^2-x^2}-\sqrt{\frac{X_n^2}{4}-x^2}-k
\geq \frac{X_n}{2}-k\geq\frac{X_n}{4}.
\end{align*}
Thus, by Lemma \ref{oscillatory integral lemma},
\begin{align*}
&\left|\int^{X_m-X_m^{-\frac{1}{3}}}_{0}e^{{\rm i}\frac{\zeta_m-\zeta_n+kx}{X_n}X_n}\la x\ra^\mu\Psi(x)dx\right|\\
\leq& CX_n^{-1}\left(\left|\left(\la x\ra^\mu\Psi\right)(X_m-X_m^{-\frac{1}{3}})\right|+\int_0^{X_m-X_m^{-\frac{1}{3}}}\left|\left(\la x\ra^\mu\Psi\right)^\prime(x)\right|dx\right)\\
\leq& CX_n^{-1}\left(X_m^\mu\left|\Psi(X_m-X_m^{-\frac{1}{3}})\right| +\int_0^{X_m-X_m^{-\frac{1}{3}}}\left(2\la x\ra^\mu\left(J_1+J_3\right)+\mu\la x\ra^{\mu-1}\frac{ x}{\la x\ra}\left|\Psi(x)\right|\right)dx\right)\\
\leq& CX_n^{-1}X_m^\mu\left(\left|\Psi(X_m-X_m^{-\frac{1}{3}})\right| +\int_0^{X_m-X_m^{-\frac{1}{3}}}\left(J_1+J_3\right)dx\right)+C\mu X_n^{-1}\int_0^{X_m-X_m^{-\frac{1}{3}}}\la x\ra^{\mu-1}\left|\Psi(x)\right|dx.
\end{align*}
Clearly,
\begin{align*}
X_m^\mu\left|\Psi(X_m-X_m^{-\frac{1}{3}})\right| &\leq CX_m^\mu\left(X_m^2-(X_m-X_m^{-\frac{1}{3}})^2\right)^{-\frac{1}{4}}
 \left(X_n^2-(X_m-X_m^{-\frac{1}{3}})^2\right)^{-\frac{1}{4}}\\
&\leq C X_m^{-\frac{1}{3}+\mu},
\end{align*}
and 
\begin{align*}
 & \int_0^{X_m-X_m^{-\frac13}}\mu \la x\ra^{\mu-1}\left|\Psi(x)\right|dx
\leq C\left(X_m^2-(X_m-X_m^{-\frac13})^2\right)^{-\frac14}
 \left(X_n^2-(X_m-X_m^{-\frac13})^2\right)^{-\frac14}\int_0^{X_m}\mu x^{\mu-1}dx\\
\leq& C\left(X_m^2-(X_m-X_m^{-\frac13})^2\right)^{-\frac14}
 \left(X_n^2-(X_m-X_m^{-\frac13})^2\right)^{-\frac14}X_m^{\mu}\leq CX_m^{-\frac13+\mu},
\end{align*}
together with
\begin{align*}
\int_0^{X_m-X_m^{-\frac{1}{3}}}J_1dx &\leq C \int_0^{X_m-X_m^{-\frac{1}{3}}} x(X_m^2-x^2)^{-\frac{5}{4}}(X_m^2-x^2)^{-\frac{1}{4}}dx
 \leq C X_m^{-\frac{1}{3}},
\end{align*}
and
\begin{align*}
\int_0^{X_m-X_m^{-\frac{1}{3}}}J_3dx&\leq CX_m^{-1}\int_0^{X_m-X_m^{-\frac{1}{3}}}(X_m-x)^{-3}dx \leq C X_m^{-\frac{1}{3}},
\end{align*}
we obtain
$\left|\int^{X_m-X_m^{-\frac{1}{3}}}_{0}\cal{F}(x) dx\right|\leq CX_m^{-\frac{1}{3}+\mu}X_n^{-1}.$
Now we turn to the remained three terms. Since $m_0< m\le n$,
\begin{align*}
\left|\int^{X_m-X_m^{-\frac{1}{3}}}_{0} \la x\ra^\mu e^{{\rm i}kx}\psi_2^{(m)}(x)\overline{\psi_1^{(n)}(x)}dx\right|
&\leq C\int_0^{X_m-X_m^{-\frac{1}{3}}}X_m^{-2+\mu}(X_m^2-x^2)^{-\frac{1}{4}}(X_n^2-x^2)^{-\frac{1}{4}}dx\\
&\leq CX_m^{-\frac{3}{2}+\mu}X_n^{-\frac12}\le Cn^{-\frac14+\frac\mu2}.
\end{align*}
Similarly, when $m_0< m\le n$, we have
$ \left|\int^{X_m-X_m^{-\frac13}}_{0} \la x\ra^\mu e^{{\rm i}kx}\psi_1^{(m)}(x)\overline{\psi_2^{(n)}(x)}dx\right|\leq Cn^{-1+\frac\mu2} $
and
$ \left|\int^{X_m-X_m^{-\frac13}}_{0} \la x\ra^\mu e^{{\rm i}kx}\psi_2^{(m)}(x)\overline{\psi_2^{(n)}(x)}dx\right|\leq Cn^{-1+\frac\mu2}.$
Thus,
$$\left|\int^{X_m-X_m^{-\frac13}}_{0} \la x\ra^\mu e^{{\rm i}kx}h_m(x)\overline{h_n(x)}dx\right|\leq \frac{C}{n^{\frac14-\frac\mu2}}\leq \frac{C}{m^{\frac18-\frac\mu4}n^{\frac18-\frac\mu4}},\quad m_0<m\le n. $$
Case 2: $k>\frac{X_n}{4}>0$.\\
Since $m\leq n$, we have $m\leq 2k^2+1$ and $n\leq 8k^2+1$, it follows
$$
\left|\int_0^{X_m-X_m^{-\frac13}} \la x\ra^\mu e^{{\rm i}kx}h_m(x)\overline{h_n(x)}dx\right| \leq CX_m^\mu\leq CX_m^\mu\frac{m^{\frac18} n^{\frac18}}{m^{\frac18}n^{\frac18}}\leq \frac{Ck^\frac12}{m^{\frac18-\frac\mu4} n^{\frac18-\frac\mu4}}.
$$
Combining with these two cases we finish the proof.
\end{proof}
\begin{lemma}\label{lemma2.7}
If $X_n\geq2X_m$,
$$\left|\int_{X_m-X_m^{-\frac{1}{3}}}^{X_m} \la x\ra^\mu e^{{\rm i}kx}h_m(x) \overline{h_n(x)}dx\right| \leq
\frac{C}{m^{\frac{1}{4}-\frac\mu4} n^{\frac{1}{4}-\frac\mu4}}, $$
where $m_0<m\leq n$.
\end{lemma}
\begin{proof}
Firstly,
\begin{align*}
\left|\int_{X_m-X_m^{-\frac{1}{3}}}^{X_m} \cal{F}(x) dx\right| \leq&C\int_{X_m-X_m^{-\frac{1}{3}}}^{X_m} \la x\ra^\mu(X_m^2-x^2)^{-\frac{1}{4}}(X_n^2-x^2)^{-\frac{1}{4}}dx\\
\leq& C X_m^{-\frac{1}{4}+\mu} (X_n^2-X_m^2)^{-\frac{1}{4}}\int_{X_m-X_m^{-\frac{1}{3}}}^{X_m}(X_m-x)^{-\frac{1}{4}}dx\\
\leq& C X_m^{-\frac{1}{4}+\mu} (X_n^2-\frac{X_n^2}{4})^{-\frac{1}{4}}X_m^{-\frac{1}{4}}\leq C X_m^{-\frac{1}{2}+\frac\mu2} X_n^{-\frac{1}{2}+\frac\mu2}.
\end{align*}
Similarly,
$$\left|\int_{X_m-X_m^{-\frac{1}{3}}}^{X_m} \la x\ra^\mu e^{{\rm i}kx}\psi_{j_1}^{(m)}(x) \overline{\psi_{j_2}^{(n)}(x)}dx\right| \leq C X_m^{-\frac{1}{2}+\frac\mu2} X_n^{-\frac{1}{2}+\frac\mu2},\  \ j_1,j_2\in\{1,2\}.$$
Thus, we finish the proof.
\end{proof}

\begin{lemma}\label{lemma2.8}
If $X_n\geq2X_m$,
$$\left|\int_{X_m}^{X_n}\la x\ra^\mu e^{{\rm i}kx}h_m(x)\overline{h_n(x)}dx\right|\leq \frac{C}{m^{\frac{1}{8}-\frac\mu4}n^{\frac{1}{12}-\frac\mu4}}, $$
where $m_0<m\leq n$.
\end{lemma}
\begin{proof}
When $X_n>2X_{m_0}$ large enough, $X_m+X_m^{-\frac{1}{3}}\leq  \frac{X_n}{2}+1\leq\frac{3}{4}X_n$. Thus,
\begin{align*}
\left|\int^{X_m+X_m^{-\frac{1}{3}}}_{X_m}\cal{F}(x) dx\right|\leq &CX_m^\mu\int^{X_m+X_m^{-\frac{1}{3}}}_{X_m}(x^2-X_m^2)^{-\frac{1}{4}} (X_n^2-x^2)^{-\frac{1}{4}}dx\\
\leq& CX_m^{-\frac{1}{4}+\mu}\left(X_n^2-(X_m+X_m^{-\frac{1}{3}})^2\right)^{-\frac{1}{4}} \int_{X_m}^{X_m+X_m^{-\frac{1}{3}}}(x-X_m)^{-\frac{1}{4}}dx\\
\leq& CX_m^{-\frac12+\frac\mu2}X_n^{-\frac12+\frac\mu2}.
\end{align*}
By $X_n\geq2X_m$, it follows $X_n-X_n^{-\frac13}\geq\frac32X_m$. By Lemma \ref{zetaesti},
\begin{align*}
\left|\int_{X_m+X_m^{-\frac{1}{3}}}^{\frac32X_m} \cal{F}(x) dx\right| \leq& CX_m^\mu\int_{X_m+X_m^{-\frac{1}{3}}}^{\frac32X_m} (x^2-X_m^2)^{-\frac{1}{4}} (X_n^2-x^2)^{-\frac{1}{4}}e^{{\rm i}\zeta_m}dx\\
\leq&C X_m^{-\frac{1}{4}+\mu}\left(X_n^2-(X_n-X_n^{-\frac{1}{3}})^2\right)^{-\frac{1}{4}} \int_{X_m+X_m^{-\frac{1}{3}}}^{\frac32X_m} (x-X_m)^{-\frac{1}{4}} e^{-(x-X_m)}dx\\
\leq&C X_m^{-\frac{1}{4}+\mu}X_n^{-\frac{1}{6}}\int_0^\infty t^{-\frac{1}{4}}e^{-t}dt\leq C X_m^{-\frac{1}{4}+\frac\mu2}X_n^{-\frac{1}{6}+\frac\mu2}.
\end{align*}
If $x\geq \frac32X_m$, then $x-X_m\geq\frac13x$, and thus
\begin{align*}
\left|\int_{\frac32X_m}^{X_n-X_n^{-\frac{1}{3}}} \cal{F}(x) dx\right| \leq& C\int_{\frac32X_m}^{X_n-X_n^{-\frac{1}{3}}} \la x\ra^\mu(x^2-X_m^2)^{-\frac{1}{4}} (X_n^2-x^2)^{-\frac{1}{4}}e^{{\rm i}\zeta_m}dx\\
\leq&C X_m^{-\frac{1}{4}}\left(X_n^2-(X_n-X_n^{-\frac{1}{3}})^2\right)^{-\frac{1}{4}} \int_{\frac32X_m}^{X_n-X_n^{-\frac{1}{3}}} (x-X_m)^{-\frac{1}{4}+\mu} e^{-(x-X_m)}dx\\
\leq&C X_m^{-\frac{1}{4}}X_n^{-\frac{1}{6}}\int_0^\infty t^{-\frac{1}{4}+\mu}e^{-t}dt\leq C X_m^{-\frac{1}{4}}X_n^{-\frac{1}{6}}.
\end{align*}
Finally,\\
\begin{align*}
\left|\int_{X_n-X_n^{-\frac{1}{3}}}^{X_n}\cal{F}(x) dx\right| \leq&CX_n^\mu\int_{X_n-X_n^{-\frac{1}{3}}}^{X_n} (x^2-X_m^2)^{-\frac{1}{4}}(X_n^2-x^2)^{-\frac{1}{4}}dx\\
\leq&C X_n^\mu\left((X_n-X_n^{-\frac{1}{3}})^2-X_m^2\right)^{-\frac{1}{4}}X_n^{-\frac{1}{4}} \int_{X_n-X_n^{-\frac{1}{3}}}^{X_n} (X_n-x)^{-\frac{1}{4}}dx\\
\leq& CX_m^{-\frac{1}{2}}X_n^{-\frac{1}{2}+\mu}\leq CX_m^{-\frac{1}{2}}X_n^{-\frac16}.
\end{align*}
Combining with all the above estimates we have
$\left|\int_{X_m}^{X_n}\cal{F}(x)dx \right|\leq \frac{C}{m^{\frac{1}{8}-\frac\mu4}n^{\frac{1}{12}-\frac\mu4}},$
which leads to
$\left|\int_{X_m}^{X_n}\la x\ra^\mu e^{{\rm i}kx}h_m(x)\overline{h_n(x)}dx\right|\leq \frac{C}{m^{\frac{1}{8}-\frac\mu4}n^{\frac{1}{12}-\frac\mu4}}.$
\end{proof}
 Lemma \ref{lemma2.4} follows from Lemma \ref{lemma2.6}, \ref{lemma2.7} and Lemma \ref{lemma2.8}. 

\subsection{the integral on $[0, X_n]$ for the case $X_m\leq X_n\leq 2X_m$}
For the case $X_m \leq X_n \leq 2X_m$, we split the integral into
$$\int^{X_n}_0 \la x\ra^\mu e^{{\rm i}kx}h_m(x)\overline{h_n(x)}dx = \left(\int^{X_m^\frac23}_0+\int_{X_m^\frac23}^{X_m-X_m^\frac{1}{3}}+ \int_{X_m-X_m^\frac{1}{3}}^{X_n}\right) \la x\ra^\mu e^{{\rm i}kx}h_m(x)\overline{h_n(x)}dx.$$

\begin{lemma}\label{lemma2.9}
For $X_m \leq X_n \leq 2X_m$, $\displaystyle\left|\int_0^{X_m^\frac23} \la x\ra^\mu e^{{\rm i}kx}h_m(x)\overline{h_n(x)}dx\right|\leq \frac{C}{m^{\frac{1}{12}-\frac\mu4} n^{\frac{1}{12}-\frac\mu4}}$ with $C>0$,
where $m_0< m\leq n$.
\end{lemma}
The proof is simple. \\
\indent Next we estimate the integral on $[X_m^\frac23,X_m-X_m^{\frac{1}{3}}]$. We first discuss two cases for $k>0$:
1. $k>X_m^\frac13$.  2. $0< k \leq X_m^\frac13$.  The first case is simple, but the second one is much complex. We will
discuss five subcases  for the second one. For $k<0$, the proof is easy to handle.\\

\begin{lemma}\label{Xm13}
If $k>X_m^\frac13$ and $X_m \leq X_n \leq 2X_m$,
$$\left|\int_{X_m^\frac{2}{3}}^{X_m-X_m^{\frac13}}\la x\ra^\mu e^{{\rm i}kx}h_m(x)\overline{h_n(x)}dx\right| \leq \frac{Ck}{m^{\frac{1}{12}-\frac\mu4}n^{\frac{1}{12}-\frac\mu4}}.$$
\end{lemma}
\begin{proof}
If $k>X_m^\frac13$, it follows $m < k^6$ and $n < 4k^6$. By $\rm{H\ddot{o}lder}$ inequality,
$$\left|\int_{X_m^\frac{2}{3}}^{X_m-X_m^{\frac13}}\la x\ra^\mu e^{{\rm i}kx}h_m(x)\overline{h_n(x)}dx\right| \leq C X_m^\mu\frac{m^{\frac{1}{12}}n^{\frac{1}{12}}} {m^{\frac{1}{12}}n^{\frac{1}{12}}} \leq \frac{Ck}{m^{\frac{1}{12}-\frac\mu4}n^{\frac{1}{12}-\frac\mu4}}.$$
\end{proof}
\begin{lemma}\label{lemma2.10}
For $0< k \leq X_m^\frac13,  X_m\leq X_n\leq 2X_m$, if $0\leq X_n^2-X_m^2\leq kX_m^\frac{2}{3}$, then
$$ \displaystyle\left|\int_{X_m^\frac23}^{X_m-X_m^\frac{1}{3}} \la x\ra^\mu e^{{\rm i}kx}h_m(x)\overline{h_n(x)}dx\right|\leq \frac{C(k^{-1}\vee 1) }{m^{\frac16-\frac\mu4}n^{\frac16-\frac\mu4}}, $$
where $m_0< m\leq n$.
\end{lemma}

\begin{proof}
We first estimate
$$\left|\displaystyle\int_{X_m^\frac23}^{X_m-X_m^\frac{1}{3}}\cal{F}(x) dx\right| = \left|C\displaystyle\int_{X_m^\frac23}^{X_m-X_m^\frac{1}{3}} \la x\ra^\mu e^{{\rm i}(\zeta_m-\zeta_n+kx)}\Psi(x)dx\right|.$$
Since
\begin{align*}
g(X_m-X_m^\frac{1}{3})&=\frac{X_n^2-X_m^2}{\sqrt{X_n^2-(X_m-X_m^{\frac{1}{3}})^2}+\sqrt{X_m^2-(X_m-X_m^{\frac{1}{3}})^2}}-k\\
                      &\leq\frac{k X_m^\frac23}{\sqrt{X_n^2-X_m^2+X_m^\frac{4}{3}}+X_m^\frac{2}{3}}-k\leq\frac{k X_m^\frac23}{2X_m^\frac{2}{3}}-k= -\frac{k}{2},
\end{align*}
together with $g^\prime(x)\geq0$,  one obtains
$|g(x)|\geq\frac{k}{2}$ for $x\in[X_m^\frac23,X_m-X_m^\frac{1}{3}]$. Then by Lemma \ref{oscillatory integral lemma},
\begin{align*}
&\left| \int_{X_m^\frac23}^{X_m-X_m^\frac{1}{3}} \la x\ra^\mu e^{{\rm i}\frac{2}{k}(\zeta_m-\zeta_n+kx)\cdot\frac{k}{2}} (X_m^2-x^2)^{-\frac{1}{4}} (X_n^2-x^2)^{-\frac{1}{4}}\cdot f_m(x)\overline{f_n(x)}dx\right|\\
\leq& Ck^{-1}\left[\left|(\la x\ra^\mu \Psi)(X_m-X_m^{\frac{1}{3}})\right| + \int_{X_m^\frac23}^{X_m-X_m^\frac{1}{3}} \left|(\la x\ra^\mu \Psi)^\prime(x)\right|dx\right]\\
\le&  Ck^{-1}\left[ X_m^{\mu}\left|\Psi(X_m-X_m^{\frac{1}{3}})\right| + \la X_m\ra^\mu\int_{X_m^\frac23}^{X_m-X_m^\frac{1}{3}} \left|\Psi^\prime(x)\right|dx +\mu\int_{X_m^\frac23}^{X_m-X_m^\frac{1}{3}} \la x\ra^{\mu-1}\left|\Psi(x)\right|dx \right].
\end{align*}
From
\begin{align*}
\left|\Psi(X_m-X_m^{\frac{1}{3}})\right| &\leq C\left(X_m^2-(X_m-X_m^{\frac{1}{3}})^2\right)^{-\frac{1}{4}} \left(X_n^2-(X_m-X_m^{\frac{1}{3}})^2\right)^{-\frac{1}{4}} \leq C X_m^{-\frac{2}{3}},
\end{align*}
and
$
\int_{X_m^\frac23}^{X_m-X_m^\frac{1}{3}} \la x\ra^{\mu-1}\left|\Psi(x)\right|dx\le CX_m^{-\frac23+\mu}.
$
By Corollary \ref{coro2.5}, $|\Psi^\prime(x)|\le C(J_1+J_3)$. From 
\begin{align*}
\int_{X_m^\frac23}^{X_m-X_m^\frac{1}{3}}J_1dx &\leq  C \int_{X_m^\frac23}^{X_m-X_m^\frac{1}{3}} x(X_m^2-x^2)^{-\frac{5}{4}}(X_n^2-x^2)^{-\frac{1}{4}}dx \leq C X_m^{-\frac{2}{3}},
\end{align*}
together with
$
\int_{X_m^\frac23}^{X_m-X_m^\frac{1}{3}}J_3dx\leq C\int_{X_m^\frac23}^{X_m-X_m^\frac{1}{3}}\leq C X_m^{-\frac53},
$
it follows by Corollary \ref{coro2.5}
$$\left|\displaystyle\int_{X_m^\frac23}^{X_m-X_m^\frac{1}{3}} \cal{F}(x) dx\right| \leq Ck^{-1} X_m^{-\frac23+\mu}.$$
The estimates for  the rest  three terms are much simpler.  In fact, when $m > m_0$,
\begin{align*}
&\left|\int_{X_m^\frac23}^{X_m-X_m^\frac{1}{3}} \la x\ra^\mu e^{{\rm i}kx}\psi^{(m)}_2(x)\overline{\psi^{(n)}_1(x)}dx\right| \leq CX_m^{-2+\mu}\int_{X_m^\frac23}^{X_m-X_m^\frac{1}{3}}(X_m^2-x^2)^{-\frac{1}{2}}dx\\
&\leq CX_m^{-\frac52+\mu} \int_{0}^{X_m-X_m^\frac{1}{3}}(X_m-x)^{-\frac{1}{2}}dx \leq CX_m^{-\frac23+\mu}.
\end{align*}
The other two terms have same estimates. Therefore,
\begin{align*}
\left|\displaystyle\int_{X_m^\frac23}^{X_m-X_m^\frac{1}{3}}\la x\ra^\mu e^{{\rm i}kx}h_m(x)\overline{h_n(x)}dx\right|
\leq& \frac{C(k^{-1}\vee 1) }{m^{\frac16-\frac\mu4} n^{\frac16-\frac\mu4}}.
\end{align*}
\end{proof}
\begin{lemma}\label{lemma2.11}
For $0< k \leq X_m^\frac13$, $X_m\leq X_n\leq 2X_m$, if $kX_m^\frac23\leq X_n^2-X_m^2\leq kX_m^\frac56$, then
$$\displaystyle\left|\int_{X_m^\frac23}^{X_m-X_m^\frac{1}{3}} \la x\ra^\mu e^{{\rm i}kx}h_m(x)\overline{h_n(x)}dx\right|\leq \frac{C(k^{-1}\vee 1)}{m^{\frac{1}{12}-\frac\mu4} n^{\frac{1}{12}-\frac\mu4}},$$
where $m_0< m\leq n$.
\end{lemma}

\begin{proof}
When $kX_m^\frac23\leq X_n^2-X_m^2\leq kX_m^\frac56$,
\begin{align*}
&\Big|\int_{X_m-X_m^{\frac13}}^{X_m-\frac{1}{32}X_m^{\frac13}}\cal{F}(x) dx\Big|\leq C   \int_{X_m-X_m^{\frac13}}^{X_m-\frac{1}{32}X_m^{\frac13}}\la x\ra^\mu(X_m^2-x^2)^{-\frac{1}{2}}dx\leq C n^{-\frac{1}{12}+\frac\mu4}m^{-\frac{1}{12}+\frac\mu4}.
\end{align*}
In the following we estimate the integral on $[X_m^\frac23, X_m-\frac{1}{32}X_m^{\frac13}]$. Since $m >m_0$ large enough,
\begin{align*}
g(X_m-\frac{1}{32}X_m^{\frac13})
\geq \frac{kX_m^\frac23}{\sqrt{kX_m^\frac56+\frac{1}{16}X_m^\frac43}+\sqrt{\frac{1}{16}X_m^\frac43}}-k
\geq \frac{kX_m^\frac23}{\frac{X_m^\frac23}{2}+\frac{X_m^\frac23}{4}}-k=\frac k3.
\end{align*}
On the other hand,
\begin{align*}
g(X_m-X_m^\frac23)
\leq \frac{kX_m^\frac56}{2\sqrt{X_m^2-(X_m-X_m^\frac23)^2}}-k
\leq \frac{kX_m^\frac56}{2X_m^\frac56}-k=-\frac k2.
\end{align*}
We  denote $g(a)=-kX_m^{-\frac13}$, then by the monotonicity of $g(x)$, $X_m-X_m^\frac23<a<X_m-\frac{1}{32}X_m^{\frac13}$.
In the following we estimate the integral on $[a,X_m-\frac{1}{32}X_m^{\frac13}]$ firstly. A straightforward computation shows us $g^{\prime\prime}(x)>0$, therefore,
\begin{align*}
\left|g^\prime(x)\right|\geq \left|\frac{a}{\sqrt{X_n^2-a^2}}-\frac{a}{\sqrt{X_m^2-a^2}}\right|
=\frac{a(g(a)+k)}{\sqrt{X_n^2-a^2}\sqrt{X_m^2-a^2}}\geq CkX_m^{-\frac23},
\end{align*}
By  Lemma \ref{oscillatory integral lemma},
\begin{eqnarray}\label{59}
&&\Big|\int_a^{X_m-\frac{1}{32}X_m^{\frac13}}\la x\ra^\mu e^{{\rm i}\frac{\zeta_m-\zeta_n+kx}{kX_m^{-\frac23}} kX_m^{-\frac23}}(X_n^2-x^2)^{-\frac{1}{4}}(X_m^2-x^2)^{-\frac{1}{4}}f_m(x)\overline{f_n(x)}dx\Big| \nonumber \\
&\leq&Ck^{-\frac{1}{2}}X_m^{\frac13}\left[\left|(\la x\ra^\mu\Psi)(X_m-\frac{1}{32}X_m^{\frac13})\right| + \int_a^{X_m-\frac{1}{32}X_m^{\frac13}} \left|(\la x\ra^\mu\Psi)^\prime(x)\right|dx\right] \nonumber \\
&\leq&Ck^{-\frac{1}{2}}X_m^{\frac13}\bigg[X_m^\mu\left|\Psi(X_m-\frac{1}{32}X_m^{\frac13})\right| +  X_m^\mu\int_a^{X_m-\frac{1}{32}X_m^{\frac13}}\left|\Psi^\prime(x)\right|dx \nonumber \\
&&+\mu\int_a^{X_m-\frac{1}{32}X_m^{\frac13}} \la x\ra^{\mu-1}\left|\Psi(x)\right|dx\bigg].
\end{eqnarray}
We compute the right terms in (\ref{59}) one by one. Clearly,
\begin{align*}
\left|\Psi(X_m-\frac{1}{32}X_m^{\frac13})\right|&\leq C\left(X_m^2-(X_m-\frac{1}{32}X_m^{\frac13})^2\right)^{-\frac{1}{4}}\left(X_n^2-(X_m-\frac{1}{32}X_m^{\frac13})^2\right)^{-\frac{1}{4}}\leq C X_m^{-\frac23}.
\end{align*}
From 
$|\Psi^\prime(x)| \leq C(J_1+J_3)$ and 
\begin{align*}
\int_a^{X_m-\frac{1}{32}X_m^{\frac13}}J_1dx &\leq C \int_a^{X_m-\frac{1}{32}X_m^{\frac13}} x(X_m^2-x^2)^{-\frac{5}{4}}(X_m^2-x^2)^{-\frac{1}{4}}dx\\
&\leq C X_m X_m^{-\frac{3}{2}} \int_a^{X_m-\frac{1}{32}X_m^{\frac13}}(X_m-x)^{-\frac{3}{2}}dx \leq CX_m^{-\frac23},
\end{align*}
and
$
\int_a^{X_m-\frac{1}{32}X_m^{\frac13}}J_3dx\leq C X_m^{-\frac53},
$
and
$
\int_a^{X_m-\frac{1}{32}X_m^{\frac13}} \la x\ra^{\mu-1}\left|\Psi(x)\right|dx\le CX_m^{-\frac23+\mu},
$
we obtain
$$\Big|\int_a^{X_m-\frac{1}{32}X_m^{\frac13}}\la x\ra^\mu e^{{\rm i}kx}\psi_1^{(m)}(x)\overline{\psi_1^{(n)}(x)}dx\Big| \leq Ck^{-\frac{1}{2}}X_m^{-\frac13+\mu}\leq \frac{Ck^{-\frac{1}{2}}}{m^{\frac{1}{12}-\frac\mu4}n^{\frac{1}{12}-\frac\mu4}}.$$
Next we estimate the integral on $[X_m^\frac23,a]$. From $\left|g(x)\right| \geq kX_m^{-\frac13}$ and  Lemma \ref{oscillatory integral lemma}, one obtains
\begin{align*}
&\Big|\int_{X_m^\frac23}^a \la x\ra^\mu e^{{\rm i}\frac{\zeta_m-\zeta_n+kx}{kX_m^{-\frac13}} kX_m^{-\frac13}}(X_n^2-x^2)^{-\frac{1}{4}}(X_m^2-x^2)^{-\frac{1}{4}}f_m(x)\overline{f_n(x)}dx\Big|\\
\leq&Ck^{-1}X_m^{\frac13}\bigg[\left|(\la x\ra^\mu\Psi)(a)\right| + \int_{X_m^\frac23}^a \left|(\la x\ra^\mu\Psi)^\prime(x)\right|dx\bigg]\\
\leq&Ck^{-1}X_m^{\frac13}\bigg[\la a\ra^\mu\left|\Psi(a)\right| + \la a\ra^\mu\int_{X_m^\frac23}^a\left|\Psi^\prime(x)\right|dx+ \mu\int_{X_m^\frac23}^a \la x\ra^{\mu-1}\left|\Psi(x)\right|dx\bigg].
\end{align*}
Clearly,
\begin{align*}
\left|\Psi(a)\right|&\leq C\left(X_m^2-(X_m-\frac{1}{32}X_m^{\frac13})^2\right)^{-\frac{1}{4}}\left(X_n^2-(X_m-\frac{1}{32}X_m^{\frac13})^2\right)^{-\frac{1}{4}}\leq C X_m^{-\frac23}.
\end{align*}
and
$\int_{X_m^\frac23}^a \la x\ra^{\mu-1}\left|\Psi(x)\right|dx\le C X_m^{-\frac23+\mu}$.
Thus, 
$$\Big|\int_{X_m^\frac23}^a\cal{F}(x) dx\Big| \leq Ck^{-1}X_m^{-\frac13+\mu}\leq \frac{C k^{-1}}{ m^{\frac{1}{12}-\frac\mu4}n^{\frac{1}{12}-\frac\mu4}}.$$
Combining with all the estimates in this part, one obtains
$$\Big|\int_{X_m^\frac23}^{X_m-X_m^{\frac{1}{3}}} \cal{F}(x) dx\Big| \leq \frac{C(k^{-1}\vee 1)}{m^{\frac{1}{12}-\frac\mu4}n^{\frac{1}{12}-\frac\mu4}}.$$
The estimates for the remained three terms are much better. Therefore, 
$$\Big|\int_{X_m^\frac23}^{X_m-X_m^{\frac{1}{3}}} \la x\ra^\mu e^{{\rm i}kx}h_m(x)\overline{h_n(x)}dx\Big| \leq \frac{C(k^{-1}\vee 1)}{m^{\frac{1}{12}-\frac\mu4}n^{\frac{1}{12}-\frac\mu4}}.$$
\end{proof}
We delay the proofs of Lemma \ref{lemma2.12}, \ref{lemma2.13}, \ref{lemma2.14} into section \ref{S5}. 
\begin{lemma}\label{lemma2.12}
For $0<k\leq X_m^\frac13, X_m\leq X_n\leq 2X_m$, if $kX_m^\frac56 \leq X_n^2-X_m^2 \leq kX_m$, then
$$\displaystyle\left|\int_{X_m^\frac23}^{X_m-X_m^\frac{1}{3}} \la x\ra^\mu e^{{\rm i}kx}h_m(x)\overline{h_n(x)}dx\right|\leq \frac{C(k^{-1}\vee 1)}{m^{\frac{1}{12}-\frac\mu4} n^{\frac{1}{12}-\frac\mu4}},$$
where $m_0< m\leq n$.
\end{lemma}

\begin{lemma}\label{lemma2.13}
For $0<k\leq X_m^\frac13, X_m\leq X_n\leq 2X_m$, if $kX_m \leq X_n^2-X_m^2 \leq 4kX_m$, then
$$\displaystyle\left|\int_{X_m^\frac23}^{X_m-X_m^\frac{1}{3}} \la x\ra^\mu e^{{\rm i}kx}h_m(x)\overline{h_n(x)}dx\right|\leq \frac{C(k^{-1}\vee 1)}{m^{\frac{1}{12}-\frac\mu4} n^{\frac{1}{12}-\frac\mu4}},$$
where $m_0<m\leq n$.
\end{lemma}

\begin{lemma}\label{lemma2.14}
For $0< k \leq X_m^\frac13, X_m\leq X_n\leq 2X_m$, if $X_n^2-X_m^2\geq 4kX_m$, then
$$ \displaystyle\left|\int_{X_m^\frac23}^{X_m-X_m^\frac{1}{3}} \la x\ra^\mu e^{{\rm i}kx}h_m(x)\overline{h_n(x)}dx\right|\leq \frac{C(k^{-1}\vee 1) }{m^{\frac{1}{6}-\frac\mu4} n^{\frac{1}{6}-\frac\mu4}}, $$
where $m_0<m\leq n$.
\end{lemma}

From Lemma \ref{Xm13} to Lemma \ref{lemma2.14}, we have
\begin{lemma}\label{lemma2.15}
For $\forall k > 0, X_m\leq X_n\leq 2X_m$, then
$$\displaystyle\left|\int_{X_m^\frac23}^{X_m-X_m^\frac{1}{3}} \la x\ra^\mu e^{{\rm i}kx}h_m(x)\overline{h_n(x)}dx\right|\leq \frac{C(k\vee k^{-1})}{m^{\frac{1}{12}-\frac\mu4} n^{\frac{1}{12}-\frac\mu4}},$$
where $m_0<m\leq n$.
\end{lemma}

\begin{lemma}\label{lemma2.16}
For $\forall k <  0, X_m\leq X_n\leq 2X_m$, then
$$\displaystyle\left|\int_{X_m^\frac23}^{X_m-X_m^\frac{1}{3}} \la x\ra^\mu e^{{\rm i}kx}h_m(x)\overline{h_n(x)}dx\right|\leq \frac{C(|k|^{-1}\vee 1)}{m^{\frac{1}{6}-\frac\mu4 } n^{\frac{1}{6}-\frac\mu4}},$$
where $m_0<m\leq n$.
\end{lemma}
For the proof see in section \ref{S5}.\\
\indent Combining with Lemma \ref{lemma2.15} and Lemma \ref{lemma2.16}, we have
\begin{lemma}\label{lemma2.17}
For $\forall k \neq  0, X_m\leq X_n\leq 2X_m$, then
$$\displaystyle\left|\int_{X_m^\frac23}^{X_m-X_m^\frac{1}{3}} \la x\ra^\mu e^{{\rm i}kx}h_m(x)\overline{h_n(x)}dx\right|\leq \frac{C(|k|\vee |k|^{-1})}{m^{\frac{1}{12}-\frac\mu4} n^{\frac{1}{12}-\frac\mu4}},$$
where $m_0< m\leq n$.
\end{lemma}
Now we turn to the last part of integral. In fact for this part we have
\begin{lemma}\label{lemma2.18}
$\forall k \ne 0, X_m\leq X_n\leq 2X_m$, $ \displaystyle\left|\int_{X_m-X_m^\frac{1}{3}}^{X_n}\la x\ra^\mu e^{{\rm i}kx}h_m(x)\overline{h_n(x)}dx\right|\leq \frac{C}{m^{\frac{1}{12}-\frac\mu4} n^{\frac{1}{12}-\frac\mu4}},$
where $m_0< m\leq n$.
\end{lemma}
\begin{proof}
 Firstly,
\begin{align*}
\left|\int_{X_m-X_m^\frac{1}{3}}^{X_m} \cal{F}(x) dx\right| &\leq C\int_{X_m-X_m^\frac{1}{3}}^{X_m}\la x\ra^\mu(X_m^2-x^2)^{-\frac{1}{4}}(X_n^2-x^2)^{-\frac{1}{4}}dx\\
&\leq CX_m^\mu\int_{X_m-X_m^\frac{1}{3}}^{X_m}(X_m^2-x^2)^{-\frac{1}{2}}dx\\
&\leq CX_m^{-\frac{1}{2}+\mu}\int_{X_m-X_m^\frac{1}{3}}^{X_m}(X_m-x)^{-\frac{1}{2}}dx\\
&\leq CX_m^{-\frac{1}{2}+\mu}X_m^{\frac{1}{6}}\leq \frac{C}{m^{\frac{1}{12}-\frac\mu4}n^{\frac{1}{12}-\frac\mu4}}.
\end{align*}
It follows 
$ \displaystyle\left|\int_{X_m-X_m^\frac{1}{3}}^{X_m} \la x\ra^\mu e^{{\rm i}kx}h_m(x)\overline{h_n(x)}dx\right|\leq \frac{C}{m^{\frac{1}{12}-\frac\mu4}n^{\frac{1}{12}-\frac\mu4}}$. For the remained integral on $[X_m,X_n]$ 
we estimate  the integral in two different cases.\\
Case 1.$X_n - X_n^{-\frac13} \ge X_m + X_m^{-\frac13}$.
We split the integral into three parts. The first part satisfies
\begin{align*}
\left|\int^{X_m+X_m^{-\frac{1}{3}}}_{X_m} \cal{F}(x) dx\right| \leq&\int^{X_m+X_m^{-\frac{1}{3}}}_{X_m} \la x\ra^\mu (x^2-X_m^2)^{-\frac{1}{4}}(X_n^2-x^2)^{-\frac{1}{4}}dx\\
\leq& C X_m^\mu X_m^{-\frac{1}{4}} X_n^{-\frac{1}{4}} (X_n-X_m-X_m^{-\frac{1}{3}})^{-\frac{1}{4}}\int^{X_m+X_m^{-\frac{1}{3}}}_{X_m}(x-X_m)^{-\frac{1}{4}}dx\\
\leq& C X_m^\mu X_m^{-\frac{1}{4}} X_n^{-\frac{1}{4}} X_n^{\frac{1}{12}}X_m^{-\frac{1}{4}} \leq C n^{-\frac{1}{6}+\frac\mu4}m^{-\frac{1}{6}+\frac\mu4}.
\end{align*}
By Lemma \ref{zetaesti}, when $x \geq X_m+X_m^{-\frac{1}{3}}$, ${\rm i}\zeta_m \leq -(x-X_m)$, then the second part satisfies
\begin{align*}
\left|\int_{X_m+X_m^{-\frac{1}{3}}}^{X_n-X_n^{-\frac{1}{3}}} \cal{F}(x) dx\right| \leq& C\int_{X_m+X_m^{-\frac{1}{3}}}^{X_n-X_n^{-\frac{1}{3}}} \la x\ra^\mu (x^2-X_m^2)^{-\frac{1}{4}} (X_n^2-x^2)^{-\frac{1}{4}}e^{{\rm i}\zeta_m}dx\\
\leq&C (2X_m)^\mu X_m^{-\frac{1}{4}}\left(X_n^2-(X_n-X_n^{-\frac{1}{3}})^2\right)^{-\frac{1}{4}} \int_{X_m+X_m^{-\frac{1}{3}}}^{X_n-X_n^{-\frac{1}{3}}} (x-X_m)^{-\frac{1}{4}} e^{{\rm i}\zeta_m}dx\\
\leq&C (2X_m)^\mu X_m^{-\frac{1}{4}}X_n^{-\frac{1}{6}}\int_0^\infty t^{-\frac{1}{4}}e^{-t}dt\leq C n^{-\frac{1}{12}+\frac\mu4}m^{-\frac{1}{12}+\frac\mu4}.
\end{align*}
The last part satisfies
\begin{align*}
\left|\int_{X_n-X_n^{-\frac{1}{3}}}^{X_n} \cal{F}(x) dx\right| \leq&\int_{X_n-X_n^{-\frac{1}{3}}}^{X_n} \la x\ra^\mu (x^2-X_m^2)^{-\frac{1}{4}}(X_n^2-x^2)^{-\frac{1}{4}}dx\\
\leq&C (2X_m)^\mu\left((X_n-X_n^{-\frac{1}{3}})^2-X_m^2\right)^{-\frac{1}{4}}X_n^{-\frac{1}{4}} \int_{X_n-X_n^{-\frac{1}{3}}}^{X_n} (X_n-x)^{-\frac{1}{4}}dx\\
\leq& C n^{-\frac{1}{6}+\frac\mu4}m^{-\frac{1}{6}+\frac\mu4}.
\end{align*}
Thus, in Case 1,
$\Big|\int_{X_m}^{X_n} \la x\ra^\mu e^{{\rm i}kx}h_m(x)\overline{h_n(x)}dx\Big|\leq \frac{C}{m^{\frac{1}{12}-\frac\mu4}
n^{\frac{1}{12}-\frac\mu4}}.$\\
Case 2.$X_n - X_n^{-\frac{1}{3}} < X_m + X_m^{-\frac{1}{3}}$. 
In fact,
\begin{align*}
\left|\int_{X_m}^{X_n} \cal{F}(x) dx\right| \leq&\int_{X_m}^{X_n} \la x\ra^\mu(x^2-X_m^2)^{-\frac{1}{4}}(X_n^2-x^2)^{-\frac{1}{4}}dx\\
\leq&C (2X_m^\mu)X_m^{-\frac{1}{4}}X_n^{-\frac{1}{4}}\int_{X_m}^{X_n} (x-X_m)^{-\frac{1}{4}}(X_n-x)^{-\frac{1}{4}}dx\\
\leq&C (2X_m^\mu)X_m^{-\frac{1}{4}}X_n^{-\frac{1}{4}}\int_{X_m}^{X_n} (x-X_m)^{-\frac{1}{2}}dx\\
\leq&C X_m^\mu X_m^{-\frac{1}{4}}X_n^{-\frac{1}{4}}(X_n-X_m)^\frac{1}{2}\leq C n^{-\frac{1}{6}+\frac\mu4}m^{-\frac{1}{6}+\frac\mu4},
\end{align*}
where  we use the symmetric property of $(x-X_m)^{-\frac{1}{4}}(X_n-x)^{-\frac{1}{4}}$ on the interval $[X_m,X_n]$. Thus,
$\displaystyle\left|\int_{X_m}^{X_n} \la x\ra^\mu e^{{\rm i}kx}h_m(x)\overline{h_n(x)}dx\right|\leq\frac{C}{m^{\frac16-\frac\mu4}n^{\frac16-\frac\mu4}}.$ Combining with all the above estimates we complete the proof.
\end{proof}
From all lemmas in this subsection we have
\begin{lemma}\label{lemma2.20}
For $\forall k \neq  0$,
$\displaystyle\left|\int_{0}^{+\infty} \la x\ra^\mu e^{{\rm i}kx}h_m(x)\overline{h_n(x)}dx\right|\leq \frac{C(|k|^{-1}\vee |k|)}{m^{\frac{1}{12}-\frac\mu4} n^{\frac{1}{12}-\frac\mu4}}. $
\end{lemma}
From Lemma \ref{lemma2.20} and the symmetry of $h_m(x)$, we complete the proof of Lemma \ref{Indecaysection1}.

\section{Appendix}\label{S5}
\subsection{a new reducibility theorem in $L^2(\R)$}
\noindent If $\mu=0$ and the perturbation terms $\varepsilon W(\nu x, \theta)$ are analytic on $\theta$ in the equation (\ref{maineq}), we can prove the reducibility in $L^2(\R)$ instead of  $\mathcal{H}^1(\R)$. More clearly, consider 1-d quantum harmonic oscillator equation
\begin{eqnarray}\label{maineq1intro}
{\rm i}\partial_t{ \psi} &=& H_{\varepsilon}(\omega t)\psi, \ x\in\R,\\
H_{\varepsilon}(\omega t) : &= & -\partial_{xx}+x^2+\varepsilon W(\nu x, \omega t), \nonumber
 \end{eqnarray}
where  $W(\varphi, \theta)$ is defined on $\T^d\times \T^n$ and satisfies (\ref{symmetry})
and for any $\varphi\in \T^d$ and all $\alpha= (\alpha_1, \cdots, \alpha_d)$,  $\partial_{\varphi}^{\alpha}W(\varphi, \theta )$  is analytic on $\T^n_{\rho}$ and continuous on $\T^d\times \overline{\T^n_{\rho}}$, where
$0\leq |\alpha|= \alpha_1+\cdots+\alpha_d \leq d([1\vee \tau]+d+2)$ and $\tau>d-1$.
\begin{Theorem}\label{maintheorem2}
Assume that $W(\varphi, \theta)$ satisfies all the above assumptions.  There exists $\varepsilon_*>0$, such that for all $0\leq \varepsilon<\varepsilon_{*}$ there exists a closed set
$\Omega_\gamma\times \Omega_1(\varepsilon)\subset [A,B]^d\times [1,2]^n$ and for any $(\nu,\omega)\in \Omega_\gamma\times \Omega_1(\varepsilon)$ 
the linear Schr\"odinger equation  (\ref{maineq1intro})
reduces to a linear autonomous equation in $L^2(\R)$.
\end{Theorem}
\begin{remark}
The set  $\Omega_\gamma$ is defined as Theorem \ref{quantumth}, while $\Omega_1(\varepsilon)$ satisfies
$Meas\big(\Omega_1(\varepsilon)\big)\rightarrow 1$ when $ \varepsilon\rightarrow 0$.
\end{remark}

   \par Similarly, we consider 1-d quantum harmonic oscillator equation
\begin{eqnarray}\label{maineq3}
{\rm i}\partial_t{ \psi} &=& \mathcal H_{\varepsilon}(\omega t)\psi, \ x\in\R,\\
\mathcal H_{\varepsilon}(\omega t) : &= & -\partial_{xx}+x^2+\varepsilon X(x, \omega t), \nonumber
 \end{eqnarray}
where
\begin{eqnarray}\label{realform2}
X(x,\theta)= \sum\limits_{k\in \Lambda} (a_k(\theta)  \sin kx+b_k(\theta) \cos kx)
\end{eqnarray}
 with $k\in \Lambda\subset \R \setminus \{0\}$ with $|\Lambda|<\infty$, $a_k(\theta)$ and $b_k(\theta)$ are  real analytic on $\T^n_{\rho}$ and continuous on $\overline{\T^n_{\rho}}$.
\begin{Corollary}\label{coro1.50}
Assume that $X(x, \theta)$ satisfies all the above assumptions.  
There exists $\varepsilon_*>0$, such that for all $0\leq \varepsilon<\varepsilon_{*}$ there exists a closed set
$\Omega_2(\varepsilon)\subset [1,2]^n$ and for any $\omega\in \Omega_2(\varepsilon)$ the linear Schr\"odinger equation  (\ref{maineq3})
reduces to a linear autonomous equation in $L^2(\R)$.
\end{Corollary}
The above proofs are based on the KAM theorem in \cite{GT11} and Lemma \ref{Indecaysection1}. We omit the details. 
\subsection{some lemmas}  Proof of Lemma \ref{psismooth}: Recall that
$\Psi^{\infty,1}_\omega(\theta)(\sum_{a\geq 1}\xi_ah_a(x))=\sum_{a\geq 1}(M^T_\omega(\theta)\xi)_ah_a(x).$\\
(a) By  the estimate in Theorem \ref{KAM},  for $\theta\in\R^n$,
\begin{align*}
&\|\Psi^{\infty,1}_\omega(\theta)-id\|_{\mathfrak{L}(\mathcal{H}^{s'},\mathcal{H}^{s'})} = \sup_{\|u\|_{\mathcal{H}^{s'}}=1}\|(\Psi^{\infty,1}_\omega(\theta)-id)u\|_ {\mathcal{H}^{s'}}\\
=&\sup_{\|\xi\|_{l_{s'}^2}=1}\|(M^T_\omega-Id)\xi\|_ {{l_{s'}^2}} \le C(n,\beta,\iota)\varepsilon^{\frac{3}{2\beta}(\frac29 \beta-\iota)}.
\end{align*}
(b) For $b=\iota-[\iota]\in(0,1)$, $\theta_1,\theta_2\in\R^n$,
$$
\frac{\|\Psi^{\infty,1}_\omega(\theta_1)-\Psi^{\infty,1}_\omega(\theta_2)\| _{\mathfrak{L}(\mathcal{H}^{s'},\mathcal{H}^{s'})}}{\|\theta_1-\theta_2\|^b}
=\frac{1}{\|\theta_1-\theta_2\|^b}\|M^T_\omega(\theta_1)-M^T_\omega(\theta_2)\|_{\mathfrak{L}(l_{s'}^2,l_{s'}^2)}.
$$
Thus,
$$\|\Psi^{\infty,1}_\omega-id\|_{C^b(\T^n,\mathfrak{L}(\mathcal{H}^{s'},\mathcal{H}^{s'}))} =\|M^T_\omega-Id\|_{C^b(\T^n,\mathfrak{L}(l_{s'}^2,l_{s'}^2))}\le C(n,\beta,\iota)\varepsilon^{\frac{3}{2\beta}(\frac29 \beta-\iota)}.$$
(c) Denote $\la\mathcal{A}(\theta),\eta\ra u:=\sum_{a\ge1}(\la A(\theta),\eta\ra\xi)_ah_a(x)$ for $\eta\in\R^n$ where $A:=(M^T_\omega-Id)_\theta^\prime$. Note that $M^T_\omega-Id\in C^{\iota}(\T^n,\mathfrak{L}(l_{s'}^2,l_{s'}^2))$, then for $\theta\in\R^n$,
\begin{align*}
\|\mathcal{A}(\theta)\|_{\mathfrak{L}(\R^n,\mathfrak{L}(\mathcal{H}^{s'},\mathcal{H}^{s'}))} =&\sup_{\|u\|_{\mathcal{H}^{s'}}=1, \|\eta\|=1} \|\la\mathcal{A}(\theta),\eta\ra u\|_{\mathcal{H}^{s'}}\\
=&\sup_{\|\xi\|_{l_{s'}^2}=1,\|\eta\|=1} \|\la A(\theta),\eta\ra \xi\|_{l_{s'}^2}\\
=&\|A(\theta)\|_{\mathfrak{L}(\R^n,\mathfrak{L}(l_{s'}^2,l_{s'}^2))}\le \|M^T_\omega-Id\|_{C^\iota(\T^n,\mathfrak{L}(l_{s'}^2,l_{s'}^2))}.
\end{align*}
Thus $\mathcal{A}(\theta)\in\mathfrak{L}\left(\R^n,\mathfrak{L}(\mathcal{H}^{s'},\mathcal{H}^{s'})\right)$ and for given $\theta_0\in\R^n$,
\begin{align*}
&\|\Psi^{\infty,1}_\omega(\theta)-\Psi^{\infty,1}_\omega(\theta_0)-\la\mathcal{A}(\theta_0),\theta-\theta_0\ra\| _{\mathfrak{L}(\mathcal{H}^{s'},\mathcal{H}^{s'}))}\\
=&\sup_{\|u\|_{\mathcal{H}^{s'}}=1}\|(\Psi^{\infty,1}_\omega(\theta)-\Psi^{\infty,1}_\omega(\theta_0)-\la\mathcal{A}(\theta_0),\theta-\theta_0\ra)u\| _{\mathcal{H}^{s'}}\\
=&\|M^T_\omega(\theta)-M^T_\omega(\theta_0)-\la A(\theta_0),\theta-\theta_0\ra\| _{\mathfrak{L}(l_{s'}^2,l_{s'}^2)}.
\end{align*}
Since $M^T_\omega-Id\in C^1(\T^n,\mathfrak{L}(l_{s'}^2,l_{s'}^2))$, it follows 
$\Psi^{\infty,1}_\omega-id\in C^1(\T^n,\mathfrak{L}(\mathcal{H}^{s'},\mathcal{H}^{s'}))$ and $\mathcal{A}=(\Psi^{\infty,1}_\omega-id)_{\theta}^\prime$. Inductively, we can show that
$$\|\Psi^{\infty,1}_\omega-id\|_{C^{\iota}(\T^n,\mathfrak{L}(\mathcal{H}^{s'},\mathcal{H}^{s'}))} =\|M^T_\omega-Id\|_{C^{\iota}(\T^n,\mathfrak{L}(l_{s'}^2,l_{s'}^2))}\le C(n,\beta,\iota)\varepsilon^{\frac{3}{2\beta}(\frac29 \beta-\iota)}.$$
The rest is similar. \qed\\
Proof of Lemma \ref{lemma2.12}.  Since $m>m_0$ large enough, we have 
\begin{align*}
g(X_m-\frac{1}{32}X_m^\frac23)\geq\frac{kX_m^\frac56}{\sqrt{kX_m+\frac{X_m^\frac53}{16}}+\sqrt{\frac{X_m^\frac53}{16}}}-k\geq\frac{kX_m^\frac56}{\frac{X_m^\frac56}{2}+\frac{X_m^\frac56}{4}}-k=\frac k3,
\end{align*}
and
$g(\frac{X_m}{2})\leq\frac{kX_m}{2\sqrt{\frac34 X_m^2}}-k=\frac{kX_m}{\sqrt3 X_m}-k<-\frac k3.
$
Thus, if we denote $g(a)=-\frac k3$, then $\frac{1}{2}X_m<a<X_m-\frac{1}{32}X_m^\frac23$. We estimate the integral on $[a,X_m-\frac{1}{32}X_m^\frac23]$ firstly.\\
\indent Since
$
g^\prime(a)=\frac{a(g(a)+k)}{\sqrt{X_n^2-a^2}\sqrt{X_m^2-a^2}}\geq CkX_m^{-1},
$
and by Lemma \ref{oscillatory integral lemma}, we have
\begin{align*}
&\Big|\int_a^{X_m-\frac{1}{32}X_m^{\frac23}}\la x\ra^\mu e^{{\rm i}\frac{\zeta_m-\zeta_n+kx}{kX_m^{-1}} kX_m^{-1}}(X_n^2-x^2)^{-\frac{1}{4}}(X_m^2-x^2)^{-\frac{1}{4}}f_m(x)\overline{f_n(x)}dx\Big| \\
\leq&Ck^{-\frac{1}{2}}X_m^{\frac12}\bigg[\left|(\la x\ra^\mu \Psi)(X_m-\frac{1}{32}X_m^\frac23)\right| + \int_a^{X_m-\frac{1}{32}X_m^{\frac23}} \left|(\la x\ra^\mu \Psi)^\prime(x)\right|dx\bigg]\\
\leq&Ck^{-\frac{1}{2}}X_m^{\frac12}\bigg[ X_m^\mu\left|\Psi(X_m-\frac{1}{32}X_m^\frac23)\right| + X_m^\mu\int_a^{X_m-\frac{1}{32}X_m^{\frac23}}\left|\Psi^\prime(x)\right|dx\\
&+ \mu\int_a^{X_m-\frac{1}{32}X_m^{\frac23}} \la x\ra^{\mu-1}\left|\Psi(x)\right|dx\bigg].
\end{align*}
Clearly,
\begin{align*}
\left|\Psi(X_m-\frac{1}{32}X_m^{\frac23})\right|&\leq C\left(X_m^2-(X_m-\frac{1}{32}X_m^{\frac23})^2\right)^{-\frac{1}{4}}\left(X_n^2-(X_m-\frac{1}{32}X_m^{\frac23})^2\right)^{-\frac{1}{4}}\leq C X_m^{-\frac56},
\end{align*}
and
$
\int_a^{X_m-\frac{1}{32}X_m^{\frac23}} \la x\ra^{\mu-1}\left|\Psi(x)\right|dx\le CX_m^{-\frac56+\mu}.
$
From 
$|\Psi^\prime(x)| \leq C(J_1+J_3)$ and 
\begin{align*}
\int_a^{X_m-\frac{1}{32}X_m^{\frac23}}J_1dx &\leq C \int_a^{X_m-\frac{1}{32}X_m^{\frac23}} x(X_m^2-x^2)^{-\frac{5}{4}}(X_m^2-x^2)^{-\frac{1}{4}}dx\\
&\leq C X_m X_m^{-\frac{3}{2}} \int_a^{X_m-\frac{1}{32}X_m^{\frac23}}(X_m-x)^{-\frac{3}{2}}dx \leq CX_m^{-\frac56},
\end{align*}
and
$
\int_a^{X_m-\frac{1}{32}X_m^{\frac23}}J_3dx \leq  C\int_a^{X_m-\frac{1}{32}X_m^{\frac23}}\frac{(X_m^2-x^2)^{\frac{1}{4}} (X_n^2-x^2)^{-\frac{1}{4}}}{X_m(X_m-x)^3} dx
\leq CX_m^{-\frac56},
$
we obtain
$$\Big|\int_a^{X_m-\frac{1}{32}X_m^{\frac23}} \cal{F}(x) dx\Big| \leq Ck^{-\frac{1}{2}}X_m^{-\frac13+\mu}\leq \frac{C k^{-\frac{1}{2}}}{m^{\frac{1}{12}-\frac\mu4}n^{\frac{1}{12}-\frac\mu4}}.$$
Next we estimate the integral on $[X_m-\frac{1}{32}X_m^{\frac23},X_m-X_m^{\frac13}]$.  From $\left|g(x)\right| \geq \frac k3$ and Lemma \ref{oscillatory integral lemma}, one obtains 
\begin{align*}
&\Big|\int_{X_m-\frac{1}{32}X_m^{\frac23}}^{X_m-X_m^{\frac13}} \la x\ra^\mu e^{{\rm i}\frac{\zeta_m-\zeta_n+kx}{k/3} k/3}(X_n^2-x^2)^{-\frac{1}{4}}(X_m^2-x^2)^{-\frac{1}{4}}f_m(x)\overline{f_n(x)}dx\Big| \\
\leq& Ck^{-1}\left[\left|(\la x\ra^\mu \Psi)(X_m-X_m^{\frac13})\right| + \int_{X_m-\frac{1}{32}X_m^{\frac23}}^{X_m-X_m^{\frac13}} \left|(\la x\ra^\mu \Psi)^\prime(x)\right|dx\right]\\
\leq&Ck^{-1}\bigg[X_m^\mu\left|\Psi(X_m-\frac{1}{32}X_m^\frac23)\right| +  X_m^\mu\int_a^{X_m-\frac{1}{32}X_m^{\frac23}}\left|\Psi^\prime(x)\right|dx\\
&+ \mu\int_a^{X_m-\frac{1}{32}X_m^{\frac23}} \la x\ra^{\mu-1}\left|\Psi(x)\right|dx\bigg].
\end{align*}
Clearly,
$
\left|\Psi(X_m-X_m^{\frac13})\right|\leq C\left(X_m^2-(X_m-X_m^{\frac13})^2\right)^{-\frac{1}{4}}\left(X_n^2-(X_m-X_m^{\frac13})^2\right)^{-\frac{1}{4}}\leq C X_m^{-\frac23},
$
and
$\int_a^{X_m-\frac{1}{32}X_m^{\frac23}} \la x\ra^{\mu-1}\left|\Psi(x)\right|dx\le CX_m^{-\frac23+\mu}$. 
Therefore,  
$$\Big|\int_{X_m-\frac{1}{32}X_m^{\frac23}}^{X_m-X_m^{\frac13}} \la x\ra^\mu e^{{\rm i}kx}\psi_1^{(m)}(x)\overline{\psi_1^{(n)}(x)}dx\Big| \leq Ck^{-1}X_m^{-\frac23+\mu}\leq \frac{C k^{-1}}{ m^{\frac16-\frac\mu4}n^{\frac16-\frac\mu4}}.$$
A straightforward  computation shows that the integral on $[X_m^\frac23,a]$ has a better estimate. Combining with all the estimates in this part, one obtains
$\Big|\int_{X_m^\frac23}^{X_m-X_m^{\frac{1}{3}}} \cal{F}(x) dx\Big| \leq \frac{C(k^{-1}\vee 1)}{m^{\frac{1}{12}-\frac\mu4}n^{\frac{1}{12}-\frac\mu4}}$. 
It is easy to check that the estimates for the remained three terms are better. Thus,
$$\Big|\int_{X_m^\frac23}^{X_m-X_m^{\frac{1}{3}}} \la x\ra^\mu e^{{\rm i}kx}h_m(x)\overline{h_n(x)}dx\Big| \leq \frac{C(k^{-1}\vee 1)}{m^{\frac{1}{12}-\frac\mu4}n^{\frac{1}{12}-\frac\mu4}}.$$ \qed\\
Proof of Lemma \ref{lemma2.13}.
Since $m>m_0$ large enough,
\begin{align*}
g(\frac{2\sqrt2}{3}X_m)\geq\frac{kX_m}{\sqrt{4kX_m+\frac{X_m^2}{9}}+\sqrt{\frac{X_m^2}{9}}}-k\geq\frac{kX_m}{\frac{X_m}{2}+\frac{X_m}{3}}-k=\frac k5.
\end{align*}
Thus, by Lemma \ref{oscillatory integral lemma},
\begin{align*}
&\Big|\int_{\frac{2\sqrt2}{3}X_m}^{X_m-X_m^\frac13} \la x\ra^\mu e^{{\rm i}\frac{\zeta_m-\zeta_n+kx}{k/5} k/5}(X_n^2-x^2)^{-\frac{1}{4}}(X_m^2-x^2)^{-\frac{1}{4}}f_m(x)\overline{f_n(x)}dx\Big|  \\
\leq&Ck^{-1}\left[\left|(\la x\ra^\mu\Psi)(X_m-X_m^\frac13)\right| + \int_{\frac{2\sqrt2}{3}X_m}^{X_m-X_m^\frac13} \left|(\la x\ra^\mu\Psi)^\prime(x)\right|dx\right]\\
\leq&Ck^{-1}\bigg[X_m^\mu\left|\Psi(X_m-X_m^\frac13)\right| +  X_m^\mu\int_{\frac{2\sqrt2}{3}X_m}^{X_m-X_m^{\frac13}}\left|\Psi^\prime(x)\right|dx+ \mu\int_{\frac{2\sqrt2}{3}X_m}^{X_m-X_m^{\frac13}} \la x\ra^{\mu-1}\left|\Psi(x)\right|dx\bigg].
\end{align*}
Similarly, 
$\Big|\int_{\frac{2\sqrt2}{3}X_m}^{X_m-X_m^\frac13}\cal{F}(x) dx\Big| \leq \frac{C k^{-1}}{ m^{\frac16-\frac\mu4}n^{\frac16-\frac\mu4}}.$
Next we estimate the integral on $[X_m^\frac23 ,\frac{2\sqrt2}{3}X_m]$. Note that
\begin{align*}
g^\prime(x) =\frac{x(X_n^2-X_m^2)}{\sqrt{X_n^2-x^2}\sqrt{X_m^2-x^2}(\sqrt{X_n^2-x^2}+\sqrt{X_m^2-x^2})}\geq CkX_m^{-\frac43},
\end{align*}
by Lemma \ref{oscillatory integral lemma} we have
\begin{align*}
&\Big|\int_{X_m^\frac23}^{\frac{2\sqrt2}{3}X_m} \la x\ra^\mu e^{{\rm i}\frac{\zeta_m-\zeta_n+kx}{kX_m^{-\frac43}} kX_m^{-\frac43}}(X_n^2-x^2)^{-\frac{1}{4}}(X_m^2-x^2)^{-\frac{1}{4}}f_m(x)\overline{f_n(x)}dx\Big| \\
\leq& Ck^{-\frac12}X_m^{\frac23}\left[|(\la x\ra^\mu\Psi)(\frac{2\sqrt2}{3}X_m)| + \int_{X_m^\frac23}^{\frac{2\sqrt2}{3}X_m} \left|(\la x\ra^\mu\Psi)^\prime(x)\right|dx\right]\\
\leq&Ck^{-\frac12}X_m^\frac23\bigg[X_m^\mu\left|\Psi(\frac{2\sqrt2}{3}X_m)\right| + X_m^\mu\int_{X_m^\frac23}^{\frac{2\sqrt2}{3}X_m}\left|\Psi^\prime(x)\right|dx+ \mu\int_{X_m^\frac23}^{\frac{2\sqrt2}{3}X_m} \la x\ra^{\mu-1}\left|\Psi(x)\right|dx\bigg].
\end{align*}
Clearly,
\begin{align*}
\left|\Psi(\frac{2\sqrt2}{3}X_m)\right|&\leq C\left(X_m^2-(\frac{2\sqrt2}{3}X_m)^2\right)^{-\frac{1}{4}}\left(X_n^2-(\frac{2\sqrt2}{3}X_m)^2\right)^{-\frac{1}{4}}\leq C X_m^{-1},
\end{align*}
and
$
\int_{X_m^\frac23}^{\frac{2\sqrt2}{3}X_m} \la x\ra^{\mu-1}\left|\Psi(x)\right|dx\le CX_m^{-1+\mu}.
$
A straightforward computation shows that the remained term has the same estimate. Thus,
$\Big|\int_{X_m^\frac23}^{\frac{2\sqrt2}{3}X_m} \cal{F}(x) dx\Big| \leq Ck^{-\frac12}X_m^{-\frac13+\mu}\leq \frac{Ck^{-\frac12}}{m^{\frac{1}{12}-\frac\mu4}n^{\frac{1}{12}-\frac\mu4}}.$
Combining with all the estimates in this part, one obtains
$\Big|\int_{X_m^\frac23}^{X_m-X_m^{\frac{1}{3}}}\cal{F}(x) dx\Big| \leq \frac{C(k^{-1}\vee 1)}{m^{\frac{1}{12}-\frac\mu4}n^{\frac{1}{12}-\frac\mu4}}.$
Since the estimates for the remained three terms are much better, it follows 
$$\Big|\int_{X_m^\frac23}^{X_m-X_m^{\frac{1}{3}}} \la x\ra^\mu e^{{\rm i}kx}h_m(x)\overline{h_n(x)}dx\Big| \leq \frac{C(k^{-1}\vee 1)}{m^{\frac{1}{12}-\frac\mu4}n^{\frac{1}{12}-\frac\mu4}}.$$\qed

Proof of Lemma \ref{lemma2.14}.
We first estimate
$$\left|\displaystyle\int_{X_m^\frac23}^{X_m-X_m^\frac{1}{3}} \cal{F}(x) dx\right| = \left|C\displaystyle\int_{X_m^\frac23}^{X_m-X_m^\frac{1}{3}} \la x\ra^\mu e^{{\rm i}(\zeta_m-\zeta_n+kx)}\Psi(x)dx\right|.$$
From 
$
g(x)=\frac{X_n^2-X_m^2}{\sqrt{X_n^2-x^2}+\sqrt{X_m^2-x^2}}-k\geq\frac{4k X_m}{X_n+X_m}-k  \geq \frac{k}{3}
$
and Lemma \ref{oscillatory integral lemma}, we have 
\begin{align*}
&\left| \int_{X_m^\frac23}^{X_m-X_m^\frac{1}{3}} \la x\ra^\mu e^{{\rm i}\frac{(\zeta_m-\zeta_n+kx)}{k/3}k/3} (X_m^2-x^2)^{-\frac{1}{4}} (X_n^2-x^2)^{-\frac{1}{4}}\cdot f_m(x)\overline{f_n(x)}dx\right|\\
\leq& Ck^{-1}\left[\left|(\la x\ra^\mu \Psi)(X_m-X_m^{\frac{1}{3}})\right| + \int_{X_m^\frac23}^{X_m-X_m^\frac{1}{3}} \left|(\la x\ra^\mu \Psi)^\prime(x)\right|dx\right]\\
\leq&Ck^{-1}\bigg[X_m^\mu\left|\Psi(X_m-X_m^{\frac{1}{3}})\right| + X_m^\mu\int_{X_m^\frac23}^{X_m-X_m^{\frac{1}{3}}}\left|\Psi^\prime(x)\right|dx\\
&+ \mu\int_{X_m^\frac23}^{X_m-X_m^{\frac{1}{3}}} \la x\ra^{\mu-1}\left|\Psi(x)\right|dx\bigg].
\end{align*}
The remained proof is similar  as Lemma \ref{lemma2.10}.\qed\\
\noindent Proof of Lemma \ref{lemma2.16}.
We estimate
$$\left|\int_{X_m^\frac{2}{3}}^{X_m-X_m^\frac{1}{3}}
\cal{F}(x) dx\right| =\left|C\int_{X_m^\frac{2}{3}}^{X_m-X_m^\frac{1}{3}}\la x\ra^\mu e^{{\rm i}(\zeta_m-\zeta_n+kx)}\Psi(x)dx\right|$$
firstly. Since
$g(x)=\sqrt{X_n^2-x^2}-\sqrt{X_m^2-x^2}-k \geq -k = |k|$
and $g^\prime(x)>0$, by Lemma \ref{oscillatory integral lemma} we have 
\begin{align*}
&\left|\int_{X_m^\frac{2}{3}}^{X_m-X_m^\frac{1}{3}}\la x\ra^\mu e^{{\rm i}\frac{\zeta_m-\zeta_n+kx}{|k|}|k|} (X_m^2-x^2)^{-\frac{1}{4}}(X_n^2-x^2)^{-\frac{1}{4}}f_m(x)\overline{f_n(x)}dx\right|\\
\leq& \frac{C}{|k|}\left[\left|(\la x\ra^\mu \Psi)(X_m-X_m^\frac{1}{3})\right| +\int_{X_m^\frac{2}{3}}^{X_m-X_m^\frac{1}{3}}|(\la x\ra^\mu \Psi)^\prime(x)|dx\right]\\
\le& C|k|^{-1}\bigg[X_m^\mu \left|\Psi(X_m-X_m^\frac{1}{3})\right| +X_m^\mu\int_{X_m^\frac23}^{X_m-X_m^\frac13}|\Psi^\prime(x)|dx +\mu\int_{X_m^\frac23}^{X_m-X_m^\frac13}\la x\ra^{\mu-1}|\Psi(x)|dx\bigg].
\end{align*}
From 
\begin{align*}
\left|\Psi(X_m-X_m^{\frac{1}{3}})\right| &\leq C\left(X_m^2-(X_m-X_m^{\frac{1}{3}})^2\right)^{-\frac{1}{4}} \left(X_n^2-(X_m-X_m^{\frac{1}{3}})^2\right)^{-\frac{1}{4}}\leq C X_m^{-\frac23}
\end{align*}
and
$\int_{X_m^\frac23}^{X_m-X_m^\frac13}\la x\ra^{\mu-1}|\Psi(x)|dx\le CX_m^{-\frac23+\mu}$,
together with 
\begin{align*}
\int_{X_m^\frac{2}{3}}^{X_m-X_m^{\frac{1}{3}}}J_1dx &\leq  C \int_{X_m^\frac{2}{3}}^{X_m-X_m^{\frac{1}{3}}} x(X_m^2-x^2)^{-\frac{5}{4}}(X_n^2-x^2)^{-\frac{1}{4}}dx \leq C X_m^{-\frac{2}{3}},
\end{align*}
and
$
\int_{X_m^\frac{2}{3}}^{X_m-X_m^{\frac{1}{3}}}J_3dx\leq C X_m^{-\frac53},
$
we have 
$\left|\displaystyle\int_{X_m^\frac{2}{3}}^{X_m-X_m^\frac{1}{3}}\cal{F}(x) dx\right| \leq \frac{C|k|^{-1}}{ m^{\frac16-\frac\mu4}n^{\frac16-\frac\mu4}}.$
Since  the other three terms we have better estimates, it follows 
$\displaystyle\left|\int_{X_m^\frac{2}{3}}^{X_m-X_m^\frac{1}{3}} \la x\ra^\mu e^{{\rm i} kx}h_m(x)\overline{h_n(x)}dx\right|\leq \frac{C(|k|^{-1}\vee 1)}{m^{\frac16-\frac\mu4} n^{\frac16-\frac\mu4}}.$

\indent  For the following lemma we denote $d_1=\min\{\gsf,\gso\}$.
\begin{lemma}\label{Bessel}
Bessel function of third kind $\belh(z)$ satisfies the following:
\begin{alignat}{2}
\Bgs{\sbelhz{z}}&\leq  1,&\quad& z\in(-\infty,-c_0),\label{belhreal}\\
\Bgs{\sbelhz{z}}&\leq  \frac{20}{d_1} |z|^{\frac{1}{6}},&& z\in[-c_1,0),\label{belhball1}\\
\Bgs{\sbelhz{z}}&\leq  \frac{Ce^{c_2}}{d_1} \max\{|z|^{\frac{1}{6}}, |z|^{\frac56}\},&& z\in (0,c_2]{\rm i},\label{belhball2}\\
\Bgs{ \sbelhz{z}}&\leq  e^{-|z|},&& z\in(c_3,\infty){\rm i},\label{belhcomp}
\end{alignat}
\end{lemma}
where $c_0>0$, $c_1\in (0,\ 1]$, $c_2, c_3$ can be arbitrary positive numbers and $C$ is a positive constant.
\begin{proof}
As in \cite{Wastan}, we have the following equalities
\begin{align}
&\forall~z\in\C\backslash\{0\},\Re\Bgp{\nu+\frac{1}{2}}>0,\notag\\
&\bslj(z)=\frac{1}{\sqrt{\pi}\Gamma\lrp{\nu+\frac{1}{2}}}\lrp{\frac{z}{2}}^{\nu}
\int_{-1}^1(1-t^2)^{\nu-\frac{1}{2}}e^{{\rm i} z t}\,dt.\label{bslj}\\
&\bslh(z)=\frac{\rm i}{\sin(\nu\pi)}\lrp{\bslj(z)e^{-{\rm i}\nu\pi}-\bsljm(z)}.\label{bslhj}\\
&\bslh(z)=\sqrt{\frac{2}{\pi z}}\frac{e^{{\rm i}\lrp{z-\frac{\nu\pi}{2}-\frac{\pi}{4}}}}{\Gamma{\lrp{\nu+\frac{1}{2}}}}
\int_0^{\infty}e^{-t}t^{\nu-\frac{1}{2}}\lrp{1+\frac{{\rm i}t}{2z}}^{\nu-\frac{1}{2}}\,dt.\label{bslh}
\end{align}
(1) When $z\in(-\infty,-c_0),~z=-|z|$. From \eqref{bslh} we have
\begin{align*}
\Bgs{\sbelhz{z}}=\frac{1}{\gsf}\Bgs{\int_0^{\infty}e^{-t}t^{-\frac{1}{6}}\Bgp{1-\frac{{\rm i}t}{2|z|}}^{-\frac{1}{6}}\,dt}
\le\frac{1}{\gsf}\int_0^{\infty}e^{-t}t^{-\frac{1}{6}}\,dt=1.
\end{align*}
Then we obtain \eqref{belhreal}.\\
(2) When $z\in[-c_1,0),~z=-|z|$. From \eqref{bslj} we have
\begin{align*}
\Bgs{\sbelj{z}}=\frac{|z|^{\frac{5}{6}}}{2^{\frac{5}{6}}\gsf}
\Bgs{\int_{-1}^1(1-t^2)^{-\frac{1}{6}}e^{-{\rm i} | z|t}\,dt}.
\end{align*}
Since
$
\Bgs{\int_{-1}^1(1-t^2)^{-\frac{1}{6}}e^{-{\rm i}|\zeta|t}\,dt}
\le 2\int_0^1(1-t)^{-\frac{1}{6}}\,dt=\frac{12}{5},
$
it follows
\begin{equation}\label{belj1}
\Bgs{\sbelj{z}}\le\frac{3|z|^{\frac{5}{6}}}{\gsf},\quad z\in[-c_1,0).
\end{equation}
Similarly,  we have
\begin{equation}\label{belj2}
\Bgs{\sbeljm{z}}\le\frac{12|z|^{\frac{1}{6}}}{\gso},\quad z\in[-c_1,0).
\end{equation}
From  \eqref{bslhj}, \eqref{belj1} and \eqref{belj2} we obtain \eqref{belhball1}.
Similarly, we obtain \eqref{belhball2}.\\
(3)When $z\in(c_3,+\infty){\rm i},\ z={\rm i} |z|$. From  \eqref{bslh} we have
\begin{align*}
\Bgs{\sbelhz{z}}=\frac{e^{-|z|}}{\gsf}\Bgs{\int_0^{\infty}e^{-t}t^{-\frac{1}{6}}\Bgp{1+\frac{t}{2|z|}}^{-\frac{1}{6}}\,dt}
\le\frac{e^{-|z|}}{\gsf}\int_0^{\infty}e^{-t}t^{-\frac{1}{6}}\,dt=e^{-|z|}.
\end{align*}
\end{proof}
 By a straightforward computation we have
\begin{lemma}\label{zetaesti}
For $x \geq X_n$, $|\zeta_n(x)| \geq \frac{2\sqrt2}{3}X_n^\frac{1}{2}(x-X_n)^\frac{3}{2}$.
Furthermore, if we suppose $x \geq X_n+X_n^{-\frac{1}{3}}$ and $X_n > 2$,  
$|\zeta_n(x)| \geq \frac{2\sqrt2}{3}X_n^\frac{1}{3}(x-X_n) > x-X_n$.
If  $0\leq x \leq X_n$,
$|\zeta_n(x)| \geq \frac{2}{3}X_n^\frac{1}{2}(X_n-x)^\frac{3}{2}$, while 
$|\zeta_n(x)| \geq\frac{2}{3} $ if $0\leq x  \leq X_n - X_n^{-\frac{1}{3}}$. 
\end{lemma}
The next lemma is  from \cite{Stein}.
\begin{lemma}\label{oscillatory integral lemma}
Suppose $\phi$ is real-valued and smooth in $(a,b)$, $\psi$ is complex-valued, and that $|\phi^{(k)}(x)|\geq1$ for all $x\in(a,b)$. Then
$$\left|\int_a^b e^{\rm i\lambda\phi(x)}\psi(x)dx\right|\leq c_k\lambda^{-1/k}\left[|\psi(b)|+\int_a^b|\psi^\prime(x)|dx\right]$$
holds when $k\geq 2$ or  $k=1$ and $\phi^\prime(x)$ is monotonic. The bound $c_k$ is independent of $\phi$, $\psi$ and $\lambda$.
\end{lemma}

The following lemma is  standard in Fourier analysis.
\begin{lemma}\label{xishuguji1}
For any $\varphi\in \T^d$ and all $\alpha= (\alpha_1, \cdots, \alpha_d)$,  if $\partial_{\varphi}^{\alpha}W(\varphi, \theta )$  is analytic on $\T^n_{\rho}$ and continuous on $\T^d\times \overline{\T^n_{\rho}}$, where
$0\leq |\alpha|= \alpha_1+\cdots+\alpha_d \leq d([1\vee \tau]+d+2)$,  then
$$|\widehat{W}(k,l)|\leq \sup\limits_{\varphi\in \T^d}\sup\limits_{|\Im \theta|<\rho}\big|\partial_{\varphi}^{{\hat{\alpha}}}W(\varphi, \theta)\big|\frac{e^{-|l|\rho}}{\la k_1\ra^{\alpha_1}\cdots \la k_d\ra^{\alpha_d}}\leq\frac{C e^{-|l|\rho}}{\la k_1\ra^{\alpha_1}\cdots \la k_d\ra^{\alpha_d}},$$
where
$$
\hat{\alpha}_j= \left\{
\begin{array}{cc}
\alpha_j, \qquad {\rm if}\  k_j\neq 0,\\
0,\qquad {\rm if}\  k_j= 0,\\
\end{array}
\right.
$$
and $C$ is a constant which is independent of $k$ and
$\widehat{W}(k,l)=\frac{1}{(2\pi)^{d+n}} \int_{\T^{d+n}} W(\varphi, \theta) e^{-{\rm i} k\varphi-{\rm i} l\theta }d\varphi d\theta$.
\end{lemma}

Proof of Lemma \ref{6.5}.
\begin{align*}
\int_x^\infty|f(t)(\lambda-q(t))^\frac12|dt\le C\left(\int_x^\infty\frac{(q(t)-\lambda)^\frac12}{|\zeta(t)|^2}dt +\int_x^\infty\frac{q''(t)}{(q(t)-\lambda)^\frac32}dt +\int_x^\infty\frac{q'^2(t)}{(q(t)-\lambda)^\frac52}dt\right).
\end{align*}
We estimate the right terms as the following. 
Since
$\int_x^\infty\frac{(q(t)-\lambda)^\frac12}{|\zeta(t)|^2}dt=\frac{1}{|\zeta(x)|}$
and  $x>2X$, we have 
$$\displaystyle|\zeta(x)|=\int_X^x(q(t)-\lambda)^\frac12dt>C\int_{\frac23x}^x (q(t))^\frac12dt>Cx(q(\frac{2x}{3}))^\frac12>Cx(q(x))^\frac12.$$
Thus,
$\int_x^\infty\frac{(q(t)-\lambda)^\frac12}{|\zeta(t)|^2}dt\le\frac{C}{x(q(x))^\frac12}.$
From $\frac{q'(x)}{q(x)}=O\big(\frac1x\big)$ and $\frac{q^{\prime\prime}(x)}{q'(x)}=O\big(\frac1x\big)$,
\begin{align*}
   &\int_x^\infty\frac{q''(t)}{(q(t)-\lambda)^\frac32}dt +\int_x^\infty\frac{q'^2(t)}{(q(t)-\lambda)^\frac52}dt \le C\left(\int_x^\infty\frac{q''(t)}{(q(t))^\frac32}dt +\int_x^\infty\frac{q(t)q'(t)}{t(q(t))^\frac52}dt\right)\\
&\le C\int_x^\infty\frac{q'(t)}{t(q(t))^\frac32}dt\le\frac{C}{x(q(x))^\frac12}.
\end{align*}
This proves the lemma.

Proof of Lemma \ref{2.2sub}. We only give a sketch of the proof. Since the integral has a singularity $X$, we split the integral into
$$\int_0^{X^\prime}+\int_{X^\prime}^X+\int_X^{X^{\prime\prime}}+\int_{X^{\prime\prime}}^\infty=I_1+I_2+I_3+I_4,$$
where $X^\prime\ge\frac12X$ and $X^{\prime\prime}\le2X$.\\
For $I_3$, by directly integrating by parts twice, we obtain
\begin{align*}
|\zeta(x)|&=\int_X^x\frac{q'(t)(q(t)-\lambda)^\frac12}{q'(t)}dt\\
&=\frac{2(q(x)-\lambda)^\frac23}{3q'(x)}(1+\frac{2(q(x)-\lambda)q^{\prime\prime}(x)}{5q'^2(x)}-S),
\end{align*}
where
$\frac{(q(x)-\lambda)q^{\prime\prime}(x)}{q'^2(x)}=O\big(\frac{x-X}{X}\big),$
and
$S=\frac{2q'(x)}{5(q(x)-\lambda)^\frac32}\int_X^x (q(t)-\lambda)^\frac52d\big(\frac{q''(t)}{q'^3(t)}\big)=O\big(\frac{(x-X)^2}{X^2}\big),$
for $X\le x\le X^{\prime\prime}$. So we can choose a suitable $X^{\prime\prime}$ so that
$\frac{2(q(x)-\lambda)q^{\prime\prime}(x)}{5q'^2(x)}\le\frac14$, 
and $|S|$ is much smaller. Thus
\begin{align*}
-\frac{1}{\zeta^2(x)}&=\frac{9q'^2(x)}{4(q(x)-\lambda)^3} \bigg[1-\frac{4(q(x)-\lambda)q^{\prime\prime}(x)}{5q'^2(x)}+O\bigg(\frac{(x-X)^2}{X^2}\bigg)\bigg]\\
&=\frac{9q'^2(x)}{4(q(x)-\lambda)^3} -\frac{9q^{\prime\prime}(x)}{5(q(x)-\lambda)^2}+O\bigg(\frac{1}{X^2|q(x)-\lambda|}\bigg).
\end{align*}
Hence
$f(x)=O\big(\frac{1}{X^2|q(x)-\lambda|}\big)$,
and 
$I_3=O\big(\frac{1}{X^2}\int_X^{X^{\prime\prime}}\frac{1}{(q(x) -\lambda)^\frac12}dx\big)=O\big(\frac{1}{X\lambda^\frac12}\big)$. 
Similar argument can be applied to $I_2$. The estimates for $I_1, I_4$ are easy. We omit it.  This proves the lemma.\qed


\begin{thebibliography}{2018}
 \bibitem{BG}
Bambusi, D., Graffi, S.: Time quasi-periodic unbounded perturbations of Schr\"odinger operators and KAM method. Commun. Math. Phys. \textbf{219}(2), 465-480 (2001)
\bibitem{BamII}
Bambusi, D.:  Reducibility of 1-d Schr\"odinger equation with time quasiperiodic un-bounded perturbations, II. Commun. Math. Phys. \textbf{353}(1), 353-378 (2017)
\bibitem{BamI}
Bambusi, D.:  Reducibility of 1-d Schr\"odinger equation with time quasiperiodic un-bounded perturbations, I. Trans. Amer. Math. Soc., \textbf{370}(3), 1823-1865 (2018)
\bibitem{BamIII}
Bambusi, D.,  Montalto, R: Reducibility of 1-d Schr\"odinger equation with unbounded time quasiperiodic perturbations, III.  J. Math. Phys. \textbf{59}, 122702 (2018)
\bibitem{BGMR17}
Bambusi, D., Gr\'{e}bert, B., Maspero, A. and Robert, D.:  Growth of Sobolev norms for abstract linear Schr\"odinger Equations. 
To appear in Journal of the European Mathematical Society. Preprint available at arXiv:1706.09708v2 (2017)
\bibitem{BGMR18}
Bambusi, D., Gr\'{e}bert, B., Maspero, A. and Robert, D.:  Reducibility of the quantum Harmonic oscillator in d-dimensions with polynomial time dependent perturbation. Analysis \& PDEs, \textbf{11}(3), 775-799(2018)
\bibitem{BBM14}
Baldi, P.,  Berti, M., Montalto, R.: KAM for quasi-linear and fully nonlinear forced perturbations of Airy equation. Math. Ann.,  \textbf{359}(1-2), 471-536 (2014)
\bibitem{BBHM17}
Baldi, P., Berti, M., Haus, E., Montalto, R: Time quasi-periodic gravity water waves in finite depth. Inventiones Math, published online July 2018.
 \bibitem{Berti}
Berti, M.: KAM for PDEs. Boll. Unione Mat. Ital. \textbf{9}, 115-142 (2016)
 \bibitem{Berti2019}
Berti, M.: KAM Theory for Partial Differential Equations. Anal. Theory Appl., \textbf{35}(3), 235-267 (2019)
\bibitem{BBP2}
 Berti, M.,  Biasco, L.,  Procesi, M.: KAM for the reversible derivative wave equation. Arch. Rational Mech. Anal. \textbf{212},  905-955 (2014)
 \bibitem{BM16}
 Berti, M., Montalto, R.: Quasi-periodic standing wave solutions for gravity-capillary water waves.  to appear in Memoirs of the Amer. Math. Society, MEMO 891. Preprint arXiv:1602.02411v1, 2016.
\bibitem{Com87}
Combescure, M.: The quantum stability problem for time-periodic perturbations of the harmonic oscillator. Ann. Inst. H. Poincar\'e Phys. Th\'eor., \textbf{47}(1), 63-83 (1987)
\bibitem{EGK}
Eliasson, H. L., Gr\'{e}bert, B., Kuksin, S. B.: KAM for the nonlinear beam equation.  Geom. Funct. Anal. \textbf{26}(6), 1588-1715(2016)
\bibitem{EK0}
  Eliasson, L.H., Kuksin, S.B.:  On reducibility of Schr\"odinger equations with quasiperiodic in time potentials. Commun. Math. Phys. \textbf{286}(1), 125-135 (2009)
\bibitem{EK}
 Eliasson, L.H.,  Kuksin, S.B.: KAM for the nonlinear Schr\"odinger equation. Ann. of Math. \textbf{172}, 371-435 (2010)
\bibitem{EV}
Enss, V., Veselic, K.: Bound states and propagating states for time - dependent hamiltonians. Ann IHP \textbf{39}(2), 159-191 (1983)
\bibitem{FP15}
Feola, R., Procesi, M.: Quasi-periodic solutions for fully nonlinear forced reversible Schr\"odinger equations. J. Diff. Eqs.  \textbf{259}(7), 3389-3447 (2015)
\bibitem{FGMP18}
Feola, R., Giuliani, F., Montalto,R.,  Procesi,M.:  Reducibility of first order linear operators on tori via Moser's  theorem.  J. Funct. Anal. \textbf{276}, 932-970 (2019)
 \bibitem{GXY}
Geng, J.,   Xu, X.,    You, J.: An infinite dimensional KAM theorem and its application to
the two dimensional cubic Schr\"odinger equation. Adv. Math. \textbf{226}, 5361-5402 (2011)
\bibitem{GY1}
 Geng, J.,   You, J.: A KAM theorem  for one dimensional
Schr\"{o}dinger equation with periodic boundary conditions. J. Diff. Eqs. \textbf{209}, 1-56 (2005)
\bibitem{GY2}
 Geng, J.,   You, J.: A KAM theorem  for Hamiltonian partial
differential equations in higher dimensional spaces. Commun. Math.
Phys. \textbf{262}, 343-372(2006)
  \bibitem{Giu17}
 F. Giuliani, Quasi-periodic solutions for quasi-linear generalized KdV equations. J. Diff. Eqs.  \textbf{262}, 5052-5132 (2017)
  \bibitem{GP16}
 Gr\'{e}bert, B.,  Paturel, E.: KAM for the Klein Gordon equation on $\mathbb S^d$.  Boll. Unione Mat. Ital. \textbf{9}(2), 237-288 (2016)
\bibitem{GT11}
  Gr\'{e}bert, B., Thomann,  L.: KAM for the Quantum Harmonic Oscillator. Commun. Math. Phys. \textbf{307}, 383-427 (2011)
\bibitem{KLiang}
 Kappeler, T., Liang, Z.: A KAM theorem for the defocusing NLS equation with periodic boundary conditions.
J. Diff. Eqs. \textbf{252}, 4068-4113 (2012)
\bibitem{KaPo}
 Kappeler, T.,    P\"oschel, J.: {\it  KDV \& KAM}. Berlin: Springer-Verlag,  2003
\bibitem{KP}
 Kuksin, S.B.,    P\"oschel, J.: Invariant Cantor manifolds of quasi-periodic
oscillations for a nonlinear Schr\"odinger equation. Ann. of Math. \textbf{143}, 149-179 (1996)
 \bibitem{Kuk93}
Kuksin, S.B.: Nearly integrable infinite-dimensional Hamiltonian systems. {\it Lecture Notes in
Mathematics}, \textbf{1556}. Berlin: Springer-Verlag,  1993
\bibitem{Ku1}
 Kuksin, S.B.: {\it  Analysis of Hamiltonina PDEs}. Oxford: Oxford University Press,  2000
\bibitem{Ku2}
Kuksin, S.B.: A KAM theorem for equations of the Korteweg-de Vries type. Rev. Math. Math. Phys. \textbf{10}(3),  1-64 (1998)
\bibitem{LZ}
  Liang, Z.: Quasi-periodic solutions for 1D Schr\"odinger equations with the nonlinearity $|u|^{2p}u$. J. Diff. Eqs. \textbf{244},  2185-2225 (2008)
 \bibitem{LiangW19}
 Liang, Z. , Wang, Z.: Reducibility of quantum harmonic oscillator on $\R^d$ with differential and quasi-periodic in time potential. J. Diff. Eqs. \textbf{267},  3355-3395 (2019)
 \bibitem{LY1}
 Liang, Z.,   You, J.: Quasi-periodic solutions for 1D Schr\"odinger equations
with higher order nonlinearity. SIAM J. Math. Anal. \textbf{36}, 1965-1990 (2005)
\bibitem{LY10}
  Liu, J.,   Yuan, X.: Spectrum for quantum Duffing oscillator and small-divisor equation with large-variable coefficient. Comm. Pure Appl. Math. \textbf{63}(9),
1145-1172 (2010)
\bibitem{LiuYuan}
  Liu, J.,    Yuan, X.: A KAM Theorem for Hamiltonian Partial Differential Equations with Unbounded Perturbations.
Commun. Math. Phys. \textbf{307}(3), 629-673 (2011)
\bibitem{IPT05}
Iooss, G., Plotnikov, P. I., Toland, J. F.:  Standing waves on an infinitely deep perfect fluid under gravity. Arch. Ration. Mech. Anal., \textbf{177}(3), 367-478 (2005)
\bibitem{Mon17a}
 Montalto, R.: Quasi-periodic solutions of forced Kirchhoff equation. Nonlinear Differ. Equ. Appl. NoDEA, 24:9, DOI:10.1007/s00030-017-0432-3, 2017.
\bibitem{Mon17b}
Montalto, R.:  A reducibility result for a class of linear wave equations on $\T^d$. International Mathematics Research Notices, page rnx167(2017)
\bibitem{Mon18}
 Montalto, R.: On the growth of Sobolev norms for a class of linear Schr\"odinger equations on the torus with superlinear dispersion. Asymptotic Analysis.  \textbf{108}(1-2), 85-114 (2018)
\bibitem{PT01}
Plotnikov, P. I., Toland, J. F.:  Nash-Moser theory for standing water waves. Arch. Ration. Mech. Anal., \textbf{159}(1), 1-83 (2001)
\bibitem{PX}
  Procesi, M., Xu, X.: Quasi-T\"oplitz functions in KAM Theorem. SIAM J. Math. Anal. \textbf{45}, 2148 - 2181(2013)
\bibitem{Pos}
  P\"{o}schel, J.: A KAM Theorem for some nonlinear partial
differential equations. Ann. Sc. Norm. sup. Pisa CI. sci. \textbf{23},  119-148 (1996)
 \bibitem{Stein}
 Stein, Elias: Harmonic Analysis: Real-variable Methods, Orthogonality and Oscillatory Integrals. Princeton University Press, 1993
 \bibitem{T1}
 Titchmarsh, E.C.:  Eigenfunction expansions associated with second-order differential equations, Part 1, 2nd edition. Oxford: Oxford University Press, 1962
\bibitem{T2}
 Titchmarsh, E.C.: Eigenfunction expansions associated with second-order differential equations, Part 2. Oxford: Oxford University Press, 1958
 \bibitem{W90}
 Wayne,C.E.: Periodic and quasi - periodic solutions for nonlinear wave equations via KAM theory. Commun. Math. Phys. \textbf{127}, 479-528 (1990)
\bibitem{Wang08}
Wang, W. M.: Pure point spectrum of the Floquet Hamiltonian for the quantum harmonic oscillator under time quasi-periodic perturbations. Commun. Math. Phys. \textbf{277}(2),
459-496 (2008)
\bibitem{WLiang17}
Wang Z, Liang Z.: Reducibility of 1D quantum harmonic oscillator perturbed by a quasiperiodic potential with logarithmic decay. Nonlinearity, \textbf{30(4)}, 1405-1448(2017)
\bibitem{Yajima}
 Yajima, K.,  Zhang, G.: Smoothing property for Schr\"odinger equations with potential superquadratic at infinity.   Commun. Math. Phys. \textbf{221}, 573-590 (2001)
\bibitem{ZGY}
  Zhang, J.,  Gao, M.,  Yuan,  X.: KAM tori for reversible partial differential equations. Nonlinearity \textbf{24},  1189-1228 (2011)
\bibitem{Wastan}
Wastan G.N.: A Treatise on the Theory of Bessel Functions, 2nd edition.
Cambridge: Cambridge University Press, 1944.
\end{thebibliography}
\end{document}